\documentclass[11pt,oneside]{article}

\usepackage[round,colon,authoryear]{natbib} 
\usepackage[title]{appendix}

\usepackage{graphicx}
\graphicspath{ {./images/} }
\usepackage{caption}
\usepackage{subcaption}

\usepackage{multirow}
\usepackage[table,xcdraw]{xcolor}

\usepackage{amsmath}
\usepackage{amsfonts}
\usepackage{amssymb}
\usepackage{amsthm}
\usepackage{mathrsfs}
\usepackage[shortlabels]{enumitem}
\usepackage{pdfsync}

\usepackage{mathrsfs}
\usepackage{calrsfs}
\DeclareMathAlphabet{\pazocal}{OMS}{zplm}{m}{n}

\usepackage{stmaryrd}
\usepackage{comment}
\usepackage{color}
\usepackage{bm}
\usepackage{times}

\usepackage{bbm}					

\usepackage{tikz}

\bmdefine\taub{\tau}
\bmdefine\mub{\mu}
\bmdefine\lab{\lambda}
\bmdefine\varsigmab{\varsigma}

 \numberwithin{equation}{section}

\newtheorem{thm}{Theorem}[section]
\newtheorem{lem}[thm]{Lemma}
\newtheorem{prop}[thm]{Proposition}
\newtheorem{cor}[thm]{Corollary}

\theoremstyle{definition}
\newtheorem{defn}[thm]{Definition}

\theoremstyle{remark}
\newtheorem{example}[thm]{Example}

\theoremstyle{remark}
\newtheorem{rem}[thm]{Remark}

\title{
Trading Strategies Generated Pathwise
\\
by Functions of Market Weights
  \thanks{
  Research supported in part by the National Science Foundation under grant NSF-DMS-14-05210.}
}

\vspace{2mm}

\author{  
\textsc{Ioannis Karatzas} \thanks{
Department of Mathematics,  Columbia University, New York, NY 10027 (E-mail: {\it ik@math.columbia.edu}), and       \textsc{Intech} Investment Management,  One Palmer Square, Suite 441, Princeton, NJ 08542    (E-mail:    {\it ikaratzas@intechjanus.com}). 
}  
\and
\textsc{Donghan Kim} \thanks{ 
Department of Mathematics, Columbia University, New York, NY 10027 (E-mail: {\it dk2571@columbia.edu}).
}
}

\begin{document}

\maketitle

\bigskip

\begin{abstract}
\noindent
Almost twenty years ago, E.R. Fernholz introduced portfolio generating functions which can be used to construct a variety of portfolios, solely in the terms of the individual companies' market weights. I. Karatzas and J. Ruf developed recently another methodology for the functional construction of portfolios, which leads to very simple conditions for strong relative arbitrage with respect to the market. In this paper, both of these notions of functional portfolio generation are generalized in a pathwise, probability-free setting; portfolio generating functions, possibly less smooth than twice-differentiable, involve the current market weights, as well as additional bounded-variation functions related to the market weights. This generalization leads to a wider class of functionally-generated portfolios than was heretofore possible, to novel methods for dealing with the ``size'' and ``momentum'' effects, and to improved conditions for outperforming the market portfolio over suitable time-horizons. 
\end{abstract}


\smallskip

\input amssym.def
\input amssym

\smallskip

\section{Introduction}

The concept of `functionally generated portfolios' was introduced by \citet{F_generating, Fe} and has been one of the essential components of stochastic portfolio theory; see \cite{FK_survey} for an overview. Portfolios generated by appropriate functions of the individual companies' market weights have wealth dynamics which can be expressed solely in terms of these weights, and do not involve any stochastic integration. Constructing such portfolios does not require any statistical estimation of parameters, or any optimization. Completely observable quantities such as the current values of `market weights', whose temporal evolution is modeled in terms of continuous semimartingales, are the only ingredients needed for building these portfolios. Once this structure has been discerned, the mathematics underpinning its construction involves just a simple application of It\^o's rule. Then the goal is to construct such portfolios that outperform a reference portfolio, for example, the market portfolio, under appropriate structural conditions.

\smallskip

\cite{Karatzas:Ruf:2017} recently found a new way for the functional generation of trading strategies, which they call `additive generation', as opposed to Fernholz's `multiplicative generation', of portfolios. This new methodology weakens the assumptions on the market model: asset prices and market weights are continuous semimartingales, and trading strategies are constructed from `regular' functions of the semimartingales without the help of stochastic calculus. Trading strategies generated in this additive manner require simpler conditions for strong relative arbitrage with respect to the market over appropriate time horizons; see also \cite{Fernholz:Karatzas:Ruf:2018}.

\smallskip

Along a different, but related, development, \cite{F1981} showed almost 40 years ago that certain aspects of It\^o calculus can be developed `path by path', without any probability structure. Once a given function admits the pathwise property of quadratic variation/covariation along a given nested sequence of partitions over a fixed time interval of finite length, this new type of It\^o's change of variable formula can be proven by an application of Taylor expansion in a surprisingly simple way. Then \cite{Wuermli} introduced in this same setting the concept of local times and the corresponding pathwise Tanaka formula, in a pathwise sense. This allows the change of variable formula to be applied to less regular functions, by involving appropriately defined pathwise local times. These concepts of local times have been further developed recently; see \cite{PerkowskiPromel2}, \cite{Davis2018}, and \cite{Cont_Perkowski}.

\smallskip

In this paper, we generalize both additive and multiplicative functional generation of trading strategies in several ways. First, we use pathwise It\^o calculus to show that one can construct trading strategies, generated additively or multiplicatively from a given function, depending on the market weights and in a manner completely devoid of probability considerations. The only analytic structure we impose is that the market weights admit continuous covariations in a pathwise sense. Secondly, we admit generating functions that depend on an additional argument of finite variation. Introducing new arguments, other than the market weights, gives extra flexibility in constructing portfolios; see \cite{Strong}, \cite{Schied:2016}, \cite{Ruf:Xie}. We present various types of additional such arguments, to the effect that a variety of new trading strategies can be generated from a function depending on them; these strategies yield new sufficient conditions for strong relative arbitrage with respect to the market portfolio. Then, we show how to apply the pathwise Tanaka formula to construct portfolios from generating functions rougher than heretofore possible. In order to use the It\^o formula, a function needs to be at least twice-differentiable, whereas the Tanaka formula requires less smooth functions, namely, absolutely continuous. Thus, usage of the Tanaka formula broadens the class of portfolio-generating functions very considerably.

\smallskip

We also present new sufficient conditions for strong relative arbitrage via additively and multiplicatively generated trading strategies. The existing sufficient condition in \cite{Karatzas:Ruf:2017} requires the generating function to be `Lyapunov', or the corresponding `Gamma function' to be nondecreasing. By contrast, the new sufficient conditions in this paper depend on the intrinsic nondecreasing structure of the generating function itself. This new condition shows that trading strategies outperforming the market portfolio can be generated from a much richer collection of functions depending on the market weights and on an additional argument of finite variation. We provide some interesting examples of such trading strategies, and empirical analysis of them.

\medskip

\noindent
\textit{Preview} : Section~\ref{sec: 2} presents the elements of the pathwise It\^o calculus that will be needed for our purposes. Section~\ref{sec: 3} defines trading strategies and regular functions, then discusses how to generate trading strategies from regular functions in ways both additive and multiplicative. Section~\ref{sec: 4} gives sufficient conditions for such trading strategies to generate strong relative arbitrage with respect to the market. Section~\ref{sec: 5} shows methods of generating trading strategies in a similar manner as in Section~\ref{sec: 3}, but from less smoother functions with help of relevant notion of local time and Tanaka formula. Section~\ref{sec: 6} gives some examples of trading strategies generated from entropic functions and corresponding sufficient conditions for strong arbitrage. Section~\ref{sec: 7} contains empirical results of portfolios discussed in Section~\ref{sec: 6}. Finally, Section~\ref{sec: 8} concludes.

\bigskip

\bigskip

\bigskip

\section{Pathwise It\^o calculus} 
 \label{sec: 2}

In what follows, we let $X=(X_1, \cdots, X_d)'$ be a $[0, \infty)^d$-valued continuous function, representing a $d$-dimensional vector of assets whose values change over time. Each component is defined on $[0, T]$, for a fixed $T>0$, and $X_i(t)$ stands for the value of the $i$\textsuperscript{th} asset at time $t \in [0, T]$.

\smallskip

We require the components of $X$ to admit continuous covariations in the pathwise sense with respect to a given, refining sequence $(\mathbb{T}_n)_{n \in \mathbb{N}}$ of partitions of $[0, T]$. The sequence $(\mathbb{T}_n)_{n \in \mathbb{N}}$ is such that each partition is of the form $\mathbb{T}_n=\{0=t_0^{(n)} < t_1^{(n)} < \cdots < t^{n}_{N(\mathbb{T}_n)}=T\}$ for $n \in \mathbb{N}$, as well as $\mathbb{T}_1 \subset \mathbb{T}_2 \subset \cdots$, and the mesh size $||\mathbb{T}_n|| := \max_{[t_j, t_{j+1}] \in \mathbb{T}_n} |t_{j+1}-t_j|$ decreases to zero as $n \rightarrow \infty$. We fix such a sequence $(\mathbb{T}_n)_{n \in \mathbb{N}}$ of partitions for the remainder of the paper.

\smallskip

Here and below, the notation $[t_j, t_{j+1}] \in \mathbb{T}_n$ means that $t_j$ and $t_{j+1}$ are consecutive points in the partition $\mathbb{T}_n$, i.e., $t_j < t_{j+1}$, $\mathbb{T}_n \cap (t_j, t_{j+1}) = \emptyset$. Also, when we write $[t_j, t_{j+1}] \in \mathbb{T}_n$ and $t_j \leq t$ simultaneously, we set $t_{j+1} = t$ when $j$ is the biggest index satisfying $t_j \leq t$. With this notation, we present the notion of the pathwise quadratic covariation of $X$ along $(\mathbb{T}_n)_{n \in \mathbb{N}}$, as follows.

\medskip

\begin{defn}	\label{Def: quadratic cov}
	A continuous function $X=(X_1, X_2, \cdots, X_d)'$ is said to have a \textit{pathwise quadratic covariation} along a given nested sequence of partitions $(\mathbb{T}_n)_{n \in \mathbb{N}}$ of $[0, T]$, if the limit of the sequence
	\begin{equation}	\label{Def: quadratic}
		\sum_{\substack{[t_j, t_{j+1}] \in \mathbb{T}_n \\ t_j \leq t}}
		\big(X_i(t_{j+1})-X_i(t_j)\big)\big(X_k(t_{j+1})-X_k(t_j)\big),	\qquad n \in \mathbb{N}
	\end{equation}
	exists for any $t \in [0, T]$ as $n \rightarrow \infty$ and the resulting mapping, denoted by $t \mapsto \langle X_i, X_k \rangle (t)$, is real-valued and continuous for every $1 \leq i, k \leq d$. We call $\langle X_i, X_k \rangle$ the \textit{pathwise quadratic covariation} of $X_i$ and $X_k$, and define the pathwise quadratic variation of $X_i$ by $\langle X_i \rangle := \langle X_i, X_i \rangle$ as usual.  
\end{defn}

\medskip

We stress that the existence of pathwise covariations and quadratic variations for the components of $X$ depends heavily on the choice of the nested, or ``refining'' sequence $(\mathbb{T}_n)_{n \in \mathbb{N}}$ of partitions. Example 5.3.2 in \cite{Bally_Caramellino_Cont}, and the arguments following, illustrate this fact. We note also that the existence of pathwise covariations and quadratic variations is required for It\^o's formula to hold in a pathwise sense. 

\smallskip

Next, we state the original one-dimensional pathwise It\^o formula, introduced by \cite{F1981}.

\medskip

\begin{thm} [Pathwise It\^o formula for paths with quadratic variation, \cite{F1981}]
	\label{Thm : 1d Ito formula}
	Fix a continuous function $X$ which admits the quadratic variation along the given nested sequence of partitions $\mathbb{T} = (\mathbb{T}_n)_{n \geq 1}$ of $[0, T]$. Then for every $C^2(\mathbb{R}, \mathbb{R})$ function $f$, the pathwise change of variable formula
	\begin{equation} 		\label{Eq : 1d Ito formula}
		f\big(X(t)\big) - f\big(X(0)\big) = \int_0^t f'\big(X(s)\big)dX(s) + \frac{1}{2}\int_0^t f''\big(X(s)\big)d\langle X\rangle(s)
	\end{equation}
	holds for $t \in [0, T]$. Here, the F{\"o}llmer-It\^o integral is defined as the pointwise limit
	\begin{equation}		\label{Eq: 1d F-I integral}
		\int_0^t f'\big(X(s)\big)dX(s) := \lim_{n \rightarrow \infty} \sum_{\substack{[t_j, t_{j+1}] \in \mathbb{T}_n \\ t_j \leq t}} f'(X(t_j)) (X(t_{j+1})-X(t_j)),
	\end{equation}
	and the last integral of the right-hand side of \eqref{Eq : 1d Ito formula} is Lebesgue-Stieltjes integral.
\end{thm}

\medskip

We will need also the pathwise It\^o formula in a higher-dimensional setting than Theorem~\ref{Thm : 1d Ito formula}, and with an extra `input' as additional argument. For this purpose, we let $A=(A_1, A_2, \cdots, A_m)'$ be an additional vector function of finite variation and consider a $(d+m)$-dimensional function $f\big(X_1(t), \cdots, X_d(t),\allowbreak A_1(t), \cdots, A_m(t)\big)$ of time $t \in [0, T]$. We say that a given function $f : \mathbb{R}^{d+m} \rightarrow \mathbb{R}$ is in $\mathbb{C}^{j, k}(\mathbb{R}^{(d+m)}, \mathbb{R})$, if it is $j$-times continuously differentiable with respect to the first $d$ components and $k$-times continuously differentiable to the last $m$ components. We also denote by $\partial_i f$ the $i$\textsuperscript{th} partial derivative, and by $D_{\ell}f$ the $(d+\ell)$\textsuperscript{th} partial derivative of $f$.

\smallskip

We now present the following version of the pathwise It\^o formula involving both $X$ and $A$. The proof is given in the Appendix, and the idea of proof is the same as that of F{\"o}llmer's original Theorem.

\medskip

\begin{thm} [Multidimensional pathwise It\^o formula]	\label{Thm : Ito formula}
	Fix a $d$-dimensional continuous function $X$ having pathwise quadratic covariations along a given sequence of partitions $\mathbb{T} = (\mathbb{T}_n)_{n \geq 1}$ of $[0, T]$, and an $m$-dimensional continuous function $A$ of finite variation defined on $[0, T]$. Then for every $f \in \mathbb{C}^{2, 1}(\mathbb{R}^{(d+m)}, \mathbb{R})$, the pathwise change of variable formula
	\begin{align}
		f\big(X(t), A(t)\big) - f\big(X(0), A(0)\big)
		&= \int_0^t \nabla f\big(X(s), A(s)\big)dX(s) 
		+ \sum_{\ell=1}^m \int_0^t D_{\ell}f\big(X(s), A(s)\big) dA_{\ell}(s)		\nonumber
		\\
		&+ \frac{1}{2} \sum_{i,k=1}^d \int_0^t \partial^2_{i, k}  f\big(X(s), A(s)\big)d\langle X_i, X_k \rangle(s) 							\label{Eq : Ito formula}
	\end{align}
	holds for $t \in [0, T]$. Here the F{\"o}llmer-It\^o integral is defined as the pointwise limit
	\begin{equation}	\label{Eq: F-I integral}
		\int_0^t \nabla f\big(X(s), A(s)\big)dX(s) := \lim_{n \rightarrow \infty} \sum_{\substack{[t_j, t_{j+1}] \in \mathbb{T}_n \\ t_j \leq t}}
		\sum_{i=1}^d \partial_i f\big(X(t_j), A(t_j)\big) \big(X_i(t_{j+1})-X_i(t_j)\big),
	\end{equation}
	whereas the other integrals of the right-hand side of \eqref{Eq : Ito formula} are Lebesgue-Stieltjes integrals.
\end{thm}

\bigskip

\bigskip

\bigskip

\section{Trading strategies generated in pathwise sense} 
 \label{sec: 3}

As in the previous section, we consider a $[0, \infty)^d$-valued, continuous function $X=(X_1, \cdots, X_d)'$ which admits continuous covariations with respect to a refining sequence $(\mathbb{T}_n)_{n \in \mathbb{N}}$ of partitions of $[0, T]$; we also let $A=(A_1, \cdots, A_m)'$ be an additional vector function of finite variation. For the purposes of this section, the components of $X$ will denote the value processes of $d$ tradable assets, and eventually stand for the market weights in an equity market. At the same time, the components of $A$ will model the evolution of an observable, but non-tradable, quantity related to these market weights.

\smallskip

For a subset $V$ of a Euclidean space, we denote by $C([0, T], V)$ the space of continuous $V$-valued functions defined on $[0, T]$; whereas $CBV([0, T], V)$ stands for the space of those functions in $C([0, T], V)$ which are of bounded variation. With this notation, we have the following definition of trading strategy with respect to the pair $(X, A)$, in the manner of \cite{Karatzas:Ruf:2017}.

\medskip

\begin{defn}[Trading strategies]
\label{def: TS}
	For the pair $(X, A)$ of a $d$-dimensional function $X \in C([0, T], \mathbb{R}^d)$ and an $m$-dimensional function $A \in CBV([0, T], \mathbb{R}^m)$, suppose that $\vartheta=(\vartheta_1, \cdots, \vartheta_d)'$ is a $d$-dimensional function with representation 
	\begin{equation}	\label{def : vartheta}
		\vartheta_i(\cdot) = \Theta_i\big(X(\cdot), A(\cdot)\big), \qquad i=1, \cdots, d.
	\end{equation}
	Here, $\Theta = (\Theta_1, \cdots, \Theta_d)'$ is a vector of functions, for which we can define an integral $\int_0^{\cdot} \vartheta(t)dX(t) \equiv \int_0^{\cdot} \sum_{i=1}^d \vartheta_i(t)dX_i(t)$ with respect to $X$; we write $\vartheta \in \mathcal{L}(X, A)$, to express this. We shall say that $\vartheta \in \mathcal{L}(X, A)$ is a \textit{trading strategy with respect to} $X$, if it is `self-financed' in the sense that
	\begin{equation}	\label{Def: self-financing}
		V^\vartheta (\cdot;X)-V^\vartheta (0;X)=\int^{\cdot}_{0} \sum^{d}_{i=1}\vartheta_i(t)dX_i(t)
	\end{equation}
	holds. In \eqref{Def: self-financing} and in what follows,
	\begin{equation}	\label{Def: value}
		V^\vartheta (t;X):=\sum^{d}_{i=1}\vartheta_i(t)X_i(t), \qquad 0 \leq t \leq T
	\end{equation}
	denotes the value process of the strategy $\vartheta$ at time $t$.
\end{defn}

The interpretation is that $\vartheta_i(t)$ stands for the ``number of shares" invested in asset $i$ at time $t$. If $X_i(t)$ is the price of this asset, then $\vartheta_i(t)X_i(t)$ is the dollar amount invested in asset $i$ at time $t$, and $V^{\vartheta} (t;X)$ the total value of investment across all assets. ``Self-financing" means that there are neither infusions nor withdrawals of capital: gains are re-invested, losses have to be absorbed. We shall write $V^{\vartheta}(\cdot)$ instead of $V^{\vartheta}(\cdot;X)$ whenever the integrator $X$ is fixed and apparent from the context.

\smallskip

The preceding pathwise It\^o formula in Theorem~\ref{Thm : Ito formula} suggests that integrands $\vartheta \in \mathcal{L}(X, A)$ of the special form $\vartheta(t)=\nabla f\big(X(t), A(t)\big)$, for some function $f \in \mathbb{C}^{2, 1}(\mathbb{R}^{(d+m)}, \mathbb{R})$, play a very important role for integrators $X \in C([0, T], \mathbb{R}^d)$ that admit finite quadratic covariations $\langle X_i, X_j \rangle$, $1 \leq i, j \leq d$ along an appropriate nested sequence of partitions. This gives rise to the following definition.

\medskip

\begin{defn} [Admissible trading strategy]
	\label{def: AI}
	Let $X$ be a $d$-dimensional function in $C([0, T], \mathbb{R}^d)$, and $A$ an $m$-dimensional function in $CBV([0, T], \mathbb{R}^m)$. A $d$-dimensional trading strategy $\vartheta:[0, \infty) \rightarrow \mathbb{R}^d$ in $\mathcal{L}(X, A)$ is called \textit{admissible trading strategy for the pair} $(X, A)$, if there exists a function $G: \mathbb{R}^d \times \mathbb{R}^m \rightarrow \mathbb{R}$ in the space $\mathbb{C}^{2, 1}(\mathbb{R}^{(d+m)}, \mathbb{R})$, such that \eqref{def : vartheta} holds for $\Theta_i=\nabla_iG$; that is,
	\begin{equation}		\label{Def: AI}
		\vartheta(t)=\nabla G\big(X(t), A(t)\big), \qquad 0 \leq t \leq T.
	\end{equation}
\end{defn}

\medskip

If $\vartheta$ is an admissible trading strategy for $(X, A)$, the last integral of \eqref{Def: self-financing} above is interpreted as a pathwise F{\"o}llmer-It\^{o} integral in the context of Theorem~\ref{Thm : Ito formula}. In an what follows, we will define a regular function for the pair $(X, A)$, consisting of a $d$-dimensional continuous function $X$, and an $m$-dimensional function $A$ in $CBV([0, T], \mathbb{R}^m)$.

\medskip

\begin{defn} [Regular function]
	\label{def: RF}
		We say that a function $G: \mathbb{R}^d \times \mathbb{R}^m \rightarrow \mathbb{R}$ in $\mathbb{C}^{2, 1}(\mathbb{R}^{(d+m)}, \mathbb{R})$ is \textit{regular} for the pair $(X, A)$, consisting of a $d$-dimensional continuous function $X$ and of a function $A \in CBV([0, T], \mathbb{R}^m)$, if the continuous function
		\begin{equation}\label{Def: gamma}
			\Gamma^G(t):= G\big(X(0), A(0)\big)-G\big(X(t), A(t)\big) + \int^{t}_{0} \nabla G\big(X(s), A(s)\big)dX(s), \quad 0 \leq t \leq T
		\end{equation}
		has finite variation on compact intervals of $[0, T]$.
\end{defn}

\medskip

\begin{rem}
	In order to define a pathwise F{\"o}llmer-It\^o integral and be able to use the pathwise It\^o calculus, we need a sufficiently smooth (in general, at least $\mathbb{C}^{2, 1}$) function $G$, and an integrand which can be cast in the form of a derivative $\nabla G$ of this function in the manner of \eqref{Def: AI}. Thus, thanks to the above definition, we can always apply the pathwise It\^o formula (Theorem~\ref{Thm : Ito formula}) to the function $G$ as in Definition~\ref{def: RF} above, and obtain another expression for the so-called ``Gamma function" $\Gamma^G(\cdot)$ in \eqref{Def: gamma}; namely,
	\begin{equation}		\label{Eq: gamma}
		\Gamma^G(t)=-\sum_{\ell=1}^{m}\int_{0}^{t} D_{\ell} G\big(X(s), A(s)\big)dA_{\ell}(s)-\frac{1}{2}\sum_{i, k=1}^{d} \int_{0}^{t} \partial_{i, k}^2 G\big(X(s), A(s)\big)d\langle X_i, X_k \rangle (s).
	\end{equation}
	Here we recall that $D_{\ell} G\big(X(s), A(s)\big)$ and $\partial_{i, k}^2 G\big(X(s), A(s)\big)$ are, respectively, the first-order $(d+\ell)$\textsuperscript{th} partial derivative and the second-order $(i, k)$\textsuperscript{th} partial derivative of $G$ at $\big(X(s), A(s)\big)$.
	
	\medskip
	
	The difference in Definition~\ref{def: RF} here, with Definition 3.1 of \cite{Karatzas:Ruf:2017}, should be noted and stressed. In \cite{Karatzas:Ruf:2017}, the integrand $\vartheta_i$ need not be the form of `gradient' of a regular function $G$. Here, the special structure of \eqref{Def: AI} for the integrand is necessary; this is the ``price one has to pay" for being able to work in a pathwise, probability-free setting, without having to invoke the theory of rough paths.
\end{rem}

\bigskip

\subsection{Trading strategies depending on the market weights}
We place ourselves from now onward in a frictionless equity market with a fixed number $d \geq 2$ of companies. We also consider a vector of continuous functions $S=(S_1, \cdots, S_d)' \in C([0, T], [0, \infty)^d)$,  where $S_i(t)$ represents the capitalization of the $i^{\text{th}}$ company at time $t \in [0, T]$. Here we take $S_i(0)>0$ and allow $S_i(t)$ to vanish at some time $t>0$, for all $i=1, \cdots, d$; but we assume also that the total capitalization $\Sigma(t) := S_1(t)+\cdots+S_d(t)$ does not vanish at any time $t \in [0, T]$.
 
\smallskip

With these ingredients, we define another vector of continuous functions $\mu=(\mu_1, \cdots, \mu_d)'$ that consists of the companies' relative market weights
\begin{equation}	\label{Def: market weights}
	\mu_i(t) := \frac{S_i(t)}{\Sigma(t)} = \frac{S_i(t)}{S_1(t)+\cdots+S_d(t)}, \qquad t \in [0, T], \quad i =1, \cdots, d.
\end{equation}
We also assume that the components of $\mu$ admit finite quadratic covariations $\langle \mu_i, \mu_j \rangle$, $1 \leq i, j \leq d$ along a given, fixed, nested sequence $(\mathbb{T}_n)_{n \in \mathbb{N}}$ of partitions of $[0, T]$, in the manner discussed at the start of Section~\ref{sec: 2}. In what follows, we will consider only regular functions of the form $G\big(\mu(\cdot), A(\cdot)\big)$ which depend on the vector of market weights $\mu$ and on some additional function $A \in CBV([0, T], \mathbb{R}^m)$. Examples of such functions $A$ appear in \eqref{Def: quadratic variation}, \eqref{Def: covariation}.

\bigskip

\subsection{Additively generated trading strategies}
We would like now to introduce an additively-generated trading strategy, starting from a regular function in the pathwise sense. For this, we will need a result from \cite{Karatzas:Ruf:2017}. For any given function $G$ which is regular for the pair $(\mu, A)$, where $\mu$ is the vector of market weights and $A$ an appropriate function in $CBV([0, T], \mathbb{R}^m)$, we consider the vector $\vartheta$ with components 
\begin{equation}	\label{Def: vartheta2}
	\vartheta_i(\cdot):=\partial_i G\big(\mu(\cdot), A(\cdot)\big), \quad i=1, \cdots, d
\end{equation}
as in \eqref{Def: AI} of the Definition~\ref{def: AI}, and the vector of functions $\varphi = (\varphi_1, \cdots, \varphi_d)'$ with components
\begin{equation}	\label{Def: varphi}
	\varphi_i(t):=\vartheta_i(t)-Q^{\vartheta}(t)-C(0), \qquad i=1, \cdots, d, \quad 0 \leq t \leq T.
\end{equation}
Here,
\begin{equation}	\label{Def: defect of SF}
	Q^{\vartheta}(t):=V^{\vartheta}(t)-V^{\vartheta}(0)-\int_{0}^{t} \sum_{i=1}^{d} \vartheta_i(s)d\mu_i(s)
\end{equation}
is the ``defect of self-financibility" at time $t \in [0, T]$ of the integrand $\vartheta$ in \eqref{Def: vartheta2}, $V^{\varphi}(t) := \sum_{i=1}^d \vartheta_i(t)\mu_i(t)$ the ``value'' of the strategy $\varphi$ at time $t \in [0, T]$ in the manner of \eqref{Def: value}, and
\begin{equation}	\label{Def: defect of B}
	C(0):=\sum_{i=1}^{d}\partial_iG\big(\mu(0), A(0)\big)\mu_i(0)-G\big(\mu(0), A(0)\big)
\end{equation}
the ``defect of balance" at time $t=0$ for the regular function $G$. By analogy with Proposition 2.3 of \cite{Karatzas:Ruf:2017}, the vector $\varphi=(\varphi_1, \cdots, \varphi_d)'$ of \eqref{Def: varphi}, \eqref{Def: vartheta2} defines a trading strategy with respect to $\mu$.

\medskip

\begin{defn} [Additive generation]
	\label{def: AG}
	We say that the trading strategy $\varphi$ of the form \eqref{Def: varphi}, \eqref{Def: vartheta2} is \textit{additively generated} by the function $G: \mathbb{R}^d \times \mathbb{R}^m \rightarrow \mathbb{R}$, which is assumed to be regular for the pair $(X, A)$.
\end{defn}

\medskip

\begin{prop}
	\label{prop: additive generation}
	Consider the trading strategy $\varphi$, generated additively as in \eqref{Def: varphi} by a regular function $G$ for the pair $(\mu, A)$, where $\mu = (\mu_1, \cdots, \mu_d)'$ is the vector of market weights and $A \in CBV([0, T], \mathbb{R}^m)$. This strategy has value
	\begin{equation}	\label{Eq: value of ATS}
	V^{\varphi}(t)=G\big(\mu(t), A(t)\big)+\Gamma^G(t), \quad 0 \leq t \leq T
	\end{equation}
	as in Definitions~\ref{def: TS} and \ref{def: RF}, and its components can be represented, for $i=1, \cdots, d$, in the form
	\begin{align}
		\varphi_i(t)&=\partial_iG\big(\mu(t), A(t)\big)+\Gamma^G(t)+G\big(\mu(t), A(t)\big)-\sum_{j=1}^{d}\mu_j(t)\partial_jG\big(\mu(t), A(t)\big) 	\label{Eq: varphi}
		\\
		&=V^{\varphi}(t)+\partial_iG\big(\mu(t), A(t)\big)-\sum_{j=1}^{d}\mu_j(t)\partial_jG\big(\mu(t), A(t)\big). \nonumber
	\end{align}
\end{prop}

\medskip

\begin{proof}
	The proof is exactly the same as that of Proposition 4.3 of \cite{Karatzas:Ruf:2017}, if we change $G\big(\mu(t)\big)$, $D_jG\big(\mu(t)\big)$ there, into $G\big(\mu(t), A(t)\big)$, $\partial_j G\big(\mu(t), A(t)\big)$ in our present context.
\end{proof}

\medskip

The decomposition \eqref{Eq: value of ATS} suggests, that we can think of $\Gamma^G(\cdot)$ in \eqref{Def: gamma}, \eqref{Eq: gamma}, as expressing the ``cumulative earnings" of the strategy $\varphi$ of \eqref{Def: varphi}, around the ``baseline" $G\big(\mu(\cdot), A(\cdot)\big)$.

\medskip

\begin{rem}
	~
	\begin{enumerate}[label=(\roman*)]
		\item When the function $G$ in Proposition~\ref{prop: additive generation} satisfies the `balance' condition,
		\begin{equation}	\label{Def: balance}
			G\big(\mu(t), A(t)\big) = \sum_{j=1}^{d}\mu_j(t)\partial_jG\big(\mu(t), A(t)\big), \qquad 0 \leq t \leq T,
		\end{equation}
		the additively generated trading strategy $\varphi$ in \eqref{Eq: varphi} takes the considerably simpler form
		\begin{equation}	\label{Eq: balanced varphi}
			\varphi_i(t)=\partial_iG\big(\mu(t), A(t)\big)+\Gamma^G(t), \qquad i = 1, \cdots, d.
		\end{equation}
		\item For an additively generated trading strategy $\varphi$ with strictly positive value process $V^{\varphi} > 0$, the corresponding portfolio weights are defined as
		\begin{equation*}
			\pi_i(t) := \frac{\varphi_i(t)\mu_i(t)}{V^{\varphi}(t)} = \frac{\varphi_i(t)\mu_i(t)}{\sum_{i=1}^d \varphi_i(t)\mu_i(t)}, \qquad i = 1, \cdots, d,
		\end{equation*}
		or with the help of \eqref{Eq: value of ATS} and \eqref{Eq: varphi}, as
		\begin{equation}		\label{Def: pi of ATS}
			\pi_i(t) = \mu_i(t)\Bigg( 1+\frac{1}{G\big(\mu(t), A(t)\big)+\Gamma^G(t)} \Big(\partial_iG\big(\mu(t), A(t)\big)-\sum_{j=1}^d\mu_j(t)\partial_jG\big(\mu(t), A(t)\big) \Big) \Bigg).
		\end{equation}
	\end{enumerate}
\end{rem}

\bigskip

\subsection{Multiplicatively generated trading strategies}		\label{sec : 3.3}
Next, we introduce the notion of multiplicatively generated trading strategy. We suppose that a function $G: \mathbb{R}^d \times \mathbb{R}^m \rightarrow \mathbb{R}$ is regular as in Definition~\ref{def: RF} for the pair $(\mu, A)$, where $\mu$ is the vector of market weights and $A$ is some additional function in $CBV([0, T], \mathbb{R}^m)$, and that the scalar function $1/G\big(\mu(\cdot), A(\cdot)\big)$ is locally bounded. This holds, for example, if $G$ is bounded away from zero. We consider the vector function $\eta=(\eta_1, \cdots, \eta_d)'$ defined by
\begin{equation}	\label{Def: eta}
	\eta_i:=\vartheta_i \times \exp\Big(\int_{0}^{\cdot}\frac{d\Gamma^G(t)}{G\big(\mu(t), A(t))}\Big)=\partial_iG\big(\mu(\cdot), A(\cdot)\big) \times \exp\Big(\int_{0}^{\cdot}\frac{d\Gamma^G(t)}{G\big(\mu(t), A(t)\big)}\Big)
\end{equation}
in the notation of \eqref{Def: gamma}, \eqref{Def: vartheta2} for $i=1, \cdots, d$. The integral here is well-defined, as $1/G\big(\mu(\cdot), A(\cdot)\big)$ is assumed to be locally bounded. Moreover, we have $\eta \in \mathcal{L}(\mu)$, since $\vartheta \in \mathcal{L}(\mu)$ from Definition~\ref{def: TS}, and the exponential term is again a locally bounded function. As before, we turn this $\eta$ into a trading strategy $\psi=(\psi_1, \cdots, \psi_d)'$ by setting
\begin{equation}	\label{Def: psi}
	\psi_i := \eta_i -Q^{\eta}-C(0), \quad i=1, \cdots, d
\end{equation}
in the manner of \eqref{Def: varphi}, and with $Q^{\eta}$, $C(0)$ defined as in \eqref{Def: defect of SF} and \eqref{Def: defect of B}.

\medskip

\begin{defn} [Multiplicative generation]
	\label{def: MG}
	The trading strategy $\psi=(\psi_1, \cdots, \psi_d)'$ of \eqref{Def: psi}, \eqref{Def: eta} is said to be \textit{multiplicatively generated} by the function $G: \mathbb{R}^d \times \mathbb{R}^m \rightarrow \mathbb{R}$.
\end{defn}

\medskip

\begin{prop}
	\label{prop: multiplicative generation}
	Consider the trading strategy $\psi=(\psi_1, \cdots, \psi_d)'$, generated as in \eqref{Def: psi} by a given function $G: \mathbb{R}^d \times \mathbb{R}^m \rightarrow \mathbb{R}$ which is regular for $(\mu, A)$. This pair consists of the vector $\mu=(\mu_1, \cdots \mu_d)'$ of market weights, and of a suitable function $A \in CBV([0, T], \mathbb{R}^m)$ such that $1/G\big(\mu(\cdot), A(\cdot)\big)$ is locally bounded.
	
	\smallskip
	
	The value generated by this strategy is given by
	\begin{equation}	\label{Eq: value of MTS}
		V^{\psi}(t)=G\big(\mu(t), A(t)\big)\exp\Big(\int_{0}^{t}\frac{d\Gamma^G(s)}{G\big(\mu(s), A(s)\big)}\Big) > 0, \qquad 0 \leq t \leq T
	\end{equation}
	in the notation of \eqref{Def: gamma}. This strategy $\psi$ can be represented for $i=1, \cdots, d$ in the form
	\begin{equation}	\label{Eq: psi}
		\psi_i(t)=V^{\psi}(t)\bigg(1+\frac{1}{G\big(\mu(t), A(t)\big)}\Big(\partial_iG\big(\mu(t), A(t)\big)-\sum_{j=1}^{d}\mu_j(t)\partial_jG\big(\mu(t), A(t)\big)\Big)\bigg).
	\end{equation}
\end{prop}

\medskip

\begin{proof}
	We follow the argument in Proposition 4.8 of \cite{Karatzas:Ruf:2017}, using the pathwise It\^o formula instead of the standard It\^o formula for semimartingales. With the notation
	\begin{equation*}
		K(t) := \exp\Big(\int_{0}^{t}\frac{d\Gamma^G(s)}{G\big(\mu(s), A(s)\big)}\Big)
	\end{equation*}
	in \eqref{Eq: value of MTS}, the pathwise It\^o formula (Theorem~\ref{Thm : Ito formula}) yields
	\begin{alignat*}{2}
		G\big(\mu(t), A(t)\big)K(t) & =G\big(\mu(0), A(0)\big)K(0) &&+\int_{0}^{t}\sum_{i=1}^{d}\partial_iG\big(\mu(s), A(s)\big)K(s)d\mu_i(s)+\int_{0}^{t}K(s)d\Gamma^G(s)
		\\
		& &&+\int_{0}^{t}\sum_{i=0}^{m}D_iG\big(\mu(s), A(s)\big)K(s)dA_i(s) \notag
		\\
		& &&+\frac{1}{2}\int_{0}^{t}\sum_{i=1}^{d}\sum_{j=1}^{d}\partial_{i, j}^2G\big(\mu(s),	 A(s)\big)K(s)d\langle\mu_i, \mu_j\rangle(s)
		\\
		&=G\big(\mu(0), A(0)\big)K(0) &&+\int_{0}^{t}\sum_{i=1}^{d}\partial_iG\big(\mu(s), A(s)\big)K(s)d\mu_i(s)
		\\
		&=G\big(\mu(0), A(0)\big)K(0) &&+\int_{0}^{t}\sum_{i=1}^{d}\eta_i(s)d\mu_i(s)
		\\
		&=G\big(\mu(0), A(0)\big)K(0) &&+\int_{0}^{t}\sum_{i=1}^{d}\psi_i(s)d\mu_i(s), \quad \quad 0 \leq t \leq T.
	\end{alignat*}
	Here, the second equality uses the expression in \eqref{Eq: gamma}, and the last equality relies on Proposition 2.3 of \cite{Karatzas:Ruf:2017}. Since \eqref{Eq: value of MTS} holds at time zero, it follows that \eqref{Eq: value of MTS} holds at any time $t \in [0, T]$. The justification for \eqref{Eq: psi} is exactly the same as that of Proposition 4.8 in \cite{Karatzas:Ruf:2017}.
\end{proof}

\medskip

\begin{rem}
	~
	\begin{enumerate}[label=(\roman*)]
	\item The multiplicatively generated trading strategy $\psi$ in \eqref{Eq: psi} takes the far simpler form
		\begin{equation}	\label{Eq: balanced psi}
			\psi_i(t)=\partial_iG\big(\mu(t), A(t)\big)\exp\Big(\int_{0}^{t}\frac{d\Gamma^G(s)}{G\big(\mu(s), A(s)\big)}\Big), \qquad i=1, \cdots, d
		\end{equation}
		when the function $G$ in Proposition~\ref{prop: multiplicative generation} is `balanced' as in \eqref{Def: balance}.
		
	\item The portfolio weights corresponding to the multiplicatively generated trading strategy $\psi$, are similarly defined as
		\begin{equation*}
			\Pi_i(t) := \frac{\psi_i(t)\mu_i(t)}{V^{\psi}(t)} = \frac{\psi_i(t)\mu_i(t)}{\sum_{i=1}^d \psi_i(t)\mu_i(t)}, \qquad i = 1, \cdots, d;
		\end{equation*}
		and with the help of \eqref{Eq: value of MTS} and \eqref{Eq: psi}, take the form
		\begin{equation*}
			\Pi_i(t) = \mu_i(t)\bigg( 1+\frac{1}{G\big(\mu(t), A(t)\big)} \Big(\partial_iG\big(\mu(t), A(t)\big)-\sum_{j=1}^d\mu_j(t)\partial_jG\big(\mu(t), A(t)\big) \Big) \bigg).
		\end{equation*}
		For a function $G$ that satisfies the ``balance" condition \eqref{Def: balance}, this last expression simplifies to
		\begin{equation*}
			\Pi_i(t) = \mu_i(t)~\frac{\partial_iG\big(\mu(t), A(t)\big)}{G\big(\mu(t), A(t)\big)}, \qquad i = 1, \cdots, d.
		\end{equation*}
	\end{enumerate}
\end{rem}

\bigskip

\bigskip

\bigskip

\section{Sufficient conditions for strong relative arbitrage} 
 \label{sec: 4}

We consider the vector $\mu=(\mu_1, \cdots, \mu_d)'$ of market weights as in \eqref{Def: market weights}. For a given trading strategy $\varphi$ with respect to the market weights $\mu$, let us recall the value process $V^{\varphi}=\sum_{i=1}^{d}\varphi_i\mu_i$ from Definition~\ref{def: TS}. For some fixed $T_* \in (0, T]$, we say that $\varphi$ is \textit{strong relative arbitrage with respect to the market} over the time-horizon $[0, T_*]$, if we have
\begin{equation}	\label{Con: strong arb2}
V^{\varphi}(t) \geq 0, \quad \forall ~ t \in [0, T_*],
\end{equation}
along with
\begin{equation}	\label{Con: strong arb}
	V^{\varphi}(T_*) > V^{\varphi}(0).
\end{equation}

\medskip

\begin{rem}
	The notion of strong relative arbitrage defined above does not depend on any probability measure, and is slightly stricter than the existing definition of strong relative arbitrage. The classical definition involves an underlying filtered probability space, and posits that the market weights $\mu_1, \cdots, \mu_d$ should be continuous, adapted stochastic processes on this space. Also, there are two types of classical arbitrage; relative arbitrage and `strong' relative arbitrage as in Definition~4.1 of \cite{Fernholz:Karatzas:Ruf:2018}. In this old definition, an underlying probability measure is essential in defining this `weak' version of relative arbitrage. However, if we posit that $\varphi$ be strong relative arbitrage when \eqref{Con: strong arb} holds for `every' realization of $\mu$, instead of `almost sure' realization of $\mu$, the notion of strong relative arbitrage can be established without referring to any probability structure. Since we constructed trading strategies in a pathwise,  probability-free setting, the `strong' version of relative arbitrage is here a more appropriate concept of arbitrage, and we adopt the above strict definition from now on.
\end{rem}

\medskip

The value process of a trading strategy generated functionally, either additively or multiplicatively, admits a quite simple representation in terms of the generating function $G$ and the derived Gamma function $\Gamma^G$ as in \eqref{Eq: value of ATS} and \eqref{Eq: value of MTS}. This simple representation provides in turn nice sufficient conditions for strong relative arbitrage with respect to the market; for example, as in Theorem 5.1 and Theorem 5.2 of \cite{Karatzas:Ruf:2017}. In this section, we find such conditions on trading strategies generated by a regular function $G\big(\mu(\cdot), A(\cdot)\big)$, which depends not only on the vector of market weights $\mu$, but also on an additional finite-variation process $A$ related to $\mu$. We also give new sufficient conditions leading to strong relative arbitrage for both additively and multiplicatively generated trading strategies, which is different from Theorem~5.1 and Theorem~5.2 of \cite{Karatzas:Ruf:2017}.

\bigskip

Until now, we have not specified the $m$-dimensional function $A \in CBV([0, T], \mathbb{R}^m)$, so it is time to consider some plausible candidates for this function of finite variation. A first suitable candidate would be the $d$-dimensional vector
\begin{equation}	\label{Def: quadratic variation}
A=\langle\mu\rangle=\big(\langle\mu_1\rangle, \langle\mu_2\rangle, \cdots, \langle\mu_d\rangle\big)'
\end{equation}
of quadratic variation of market weights. We can also think about a more general candidate; namely, the $S_{d}^{+}$-valued covariation process of market weights. Here, $S_{d}^{+}$ is the notation for symmetric positive $d\times d$ matrices, and we will use double bracket $\langle\langle\quad \rangle\rangle$ to distinguish this $d^2$-dimensional vector from \eqref{Def: quadratic variation}: namely,
\begin{equation}	\label{Def: covariation}
A=\langle\langle\mu\rangle\rangle, \quad (A)_{i, j}=\langle \mu_i , \mu_j \rangle \qquad 1 \leq i, j \leq d.
\end{equation}
The advantage of choosing $A$ as in \eqref{Def: covariation}, is that we can match the integrators of the two integrals in \eqref{Eq: gamma}, and the resulting expression for $\Gamma^G(\cdot)$ can then be cast as one integral.

\smallskip

There are many other functions of finite variation which can be candidates for the process $A$. We list some examples below:
\begin{enumerate}
	
	\item The moving average $\bar{\mu}$ of $\mu$ defined by
	\begin{equation*}
		\bar{\mu}_i(t) := \begin{cases} \frac{1}{\delta}\int_{0}^{t}\mu_i(s)ds+\frac{1}{\delta}\int_{t-\delta}^{0}\mu_i(0)ds , & t \in [0, \delta),
		\\
		\frac{1}{\delta}\int_{t-\delta}^{t}\mu_i(s)ds, & t \in [\delta, T], 
		\end{cases} \quad i = 1, \cdots, d.
	\end{equation*}
	
	\item The running maximum $\mu^*$ of the market weights with the components $\mu^*_i(t) := \max_{0 \leq s \leq t}\mu_i(s)$, and the running minimum $\mu_*$ of the market weights with the components $\mu_{*i}(t) := \min_{0 \leq s \leq t}\mu_i(s)$ for $i = 1, \cdots, d$.

	\item The `pathwise local time' $L_{\cdot}^{\mu_{(i)}-\mu_{(k)}}(0)$ of $\mu_{(i)}-\mu_{(k)} \geq 0 $ at the origin, for $1 \leq i < k \leq d$, which is defined in Section~\ref{sec: 5}. We call this process the ``collision local time" of order $k-i+1$~(the number of particles involved in the collison), for the ranked market weights
	\begin{equation*}
		\mu_{(1)} := \max_j \mu_j \geq \mu_{(2)} \geq \cdots \geq \mu_{(d)} =: \min_j \mu_j.
	\end{equation*} 
\end{enumerate}
Since the vectors $\bar{\mu}, \mu^*$, and $\mu_*$, defined above, are $d$-dimensional, $m=d$ holds for these choices of $A$. For the choice of $\frac{1}{2}d(d-1)$-dimensional vector $\Lambda$ with the components $(\Lambda)_{i, k} := L^{\mu_{(i)}-\mu_{(k)}}$, the dimension $m$ of $A$ is $\frac{1}{2}n(n-1)$. Empirical results using the moving average $\bar{\mu}$ can be found in Section 3 of \cite{Schied:2016}. The collision local times $(\Lambda)_{i, k}$ always appear when we deal with function of ranked market weights, as in Example 3.9 of \cite{Karatzas:Ruf:2017}.

\smallskip

We first consider conditions leading to strong relative arbitrage with respect to the market with general $A$ as the third input of generating function $G$. Then we present some examples of $G$ with specific finite variation function $A$ chosen from among the above candidates, and continue with empirical results regarding these examples.

\bigskip

\subsection{Additively generated strong relative arbitrage}	\label{sec : 4.1}

We start with a condition leading to additively generated strong arbitrage, which is similar to Theorem 5.1 of \cite{Karatzas:Ruf:2017}.

\medskip

\begin{thm} [Additively generated strong relative arbitrage when $\Gamma^G$ is nondecreasing]
	\label{thm: AGRA}
	Fix a function $G: \mathbb{R}^d \times \mathbb{R}^m \rightarrow [0, \infty)$ which is regular for the pair $(\mu, A)$, and such that the function $\Gamma^G(\cdot)$ in \eqref{Def: gamma} or \eqref{Eq: gamma} is nondecreasing. Here, $\mu$ is the vector of market weights and $A$ is some $m$-dimensional function in $CBV([0, T], \mathbb{R}^m)$, as before.
	
	\smallskip
	
	For some real number $T_*>0$, suppose that
	\begin{equation}	\label{Con: additive arb}
		\Gamma^G(T_*) > G\big(\mu(0), A(0)\big)
	\end{equation}
	holds. Then the trading strategy $\varphi$, additively generated by the regular function $G$ as in Definition~\ref{def: AG}, is strong arbitrage relative to the market over every time horizon $[0, t]$ with $T_* \leq t \leq T$.
\end{thm}

\medskip

\begin{proof}
	Since $\Gamma^G(\cdot)$ is nondecreasing, we obtain $V^{\varphi}(t)=G\big(\mu(t), A(t)\big)+\Gamma^G(t) \geq \Gamma^G(0)=0$ for every $t \in [0, T_*]$ from \eqref{Eq: value of ATS}. We also have $V^{\varphi}(t)=G\big(\mu(t), A(t)\big)+\Gamma^G(t) \geq \Gamma^G(T_*)>G\big(\mu(0), A(0)\big)=V^{\varphi}(0)$ for all $t \in [T_*, T]$. The last equality holds because $\Gamma^G(0)=0$.
\end{proof}

\medskip

\begin{rem}
	If we choose $A=\langle\langle\mu\rangle\rangle$ as in \eqref{Def: covariation}, then from \eqref{Eq: gamma}, the function $\Gamma^G(\cdot)$ is nondecreasing when
	\begin{equation*}
		-\sum_{i, k=1}^{d}\int_{0}^{\cdot} \Big( D^{1}_{(i, k)} + \frac{1}{2} \partial_{i, k}^2 \Big) G\big(\mu(s), \langle\langle \mu \rangle\rangle(s)\big) d\langle \mu_i, \mu_j \rangle(s)
	\end{equation*}
	is nondecreasing. Here, $D^{1}_{(i, k)}$ denotes the first-order partial derivative operator with respect to the $(i, k)$th entry of $\langle\langle \mu \rangle\rangle$. Also, we substitute from \eqref{Def: covariation}, \eqref{Eq: gamma} into \eqref{Con: additive arb} to obtain the more explicit form
	\begin{equation}			\label{Eq: additive Gamma with QV}
		-\sum_{i, k=1}^{d}\int_{0}^{T_*}\Big( D^{1}_{(i, k)} + \frac{1}{2} \partial_{i, k}^2 \Big) G\big(\mu(s), \langle\langle \mu \rangle\rangle(s)\big) d\langle \mu_i, \mu_j \rangle(s)		
		>G\big(\mu(0), \langle\langle \mu \rangle\rangle(0)\big)
	\end{equation}
	of the condition \eqref{Con: additive arb} for strong relative arbitrage. Thus, unlike the situation of Theorem 3.7 in \cite{Karatzas:Ruf:2017}, we can have a nondecreasing $\Gamma^G$ and a chance for effecting strong relative arbitrage, even without the `concavity' of $G$ in $\mu$.
\end{rem}

\medskip

\begin{rem}
	Let us assume that the arguments $\mu$ and $A$ are `additively separated' in the function $G$. By this we mean, that there exist two regular functions $K$ and $H$ with the property that $K$ depends only on $\mu(t)$ and $H$ depends on $A(t)$, and such that 
	\begin{equation} 	\label{Eq: seperated case}
		G\big(\mu(t), A(t)\big)=K\big(\mu(t)\big)+H\big(A(t)\big), \quad \forall \ t \in [0, T]
	\end{equation}
	holds. Then, we get $\partial_{i, k}^2 G\big(\mu(t), A(t)\big)=\partial_{i, k}^2 K\big(\mu(t)\big)$ and $D_{\ell}G\big(\mu(t), A(t)\big)=D_{\ell}H\big(A(t)\big)$. Substituting these expressions into \eqref{Eq: gamma}, we obtain
	\begin{equation}	\label{Eq: gamma in separated case}
		\Gamma^G(T_*) =-\sum_{\ell=1}^m \int_{0}^{T_*}D_{\ell}H\big(A(s)\big)dA_{\ell}(s)
		-\frac{1}{2}\sum_{i, k=1}^{d}\int_{0}^{T_*} \partial_{i, k}^2 K\big(\mu(s)\big) d\langle \mu_i, \mu_k \rangle(s)
	\end{equation}
	and, from \eqref{Eq: value of ATS} of Proposition~\ref{prop: additive generation}, the relative value process of the additively generated trading strategy $\varphi$ by $G$ can be expressed as
	\begin{equation}	\label{Eq: value in separated case}
		V^{\varphi}(T_*)=K\big(\mu(T_*)\big)+H\big(A(T_*)\big)+\Gamma^G(T_*), \qquad V^{\varphi}(0)=K\big(\mu(0)\big)+H\big(A(0)\big).
	\end{equation}
	After substituting \eqref{Eq: gamma in separated case}, \eqref{Eq: value in separated case} into \eqref{Con: strong arb} and rearranging terms in such a manner that the left-hand side contains only terms involving $K$, the strong arbitrage condition \eqref{Con: strong arb} takes the form
	\begin{equation} 	\label{Ineq: separated case}
		K\big(\mu(T_*)\big)-K\big(\mu(0)\big) -\frac{1}{2}\sum_{i, k=1}^{d}\int_{0}^{T_*}\partial_{i, k}^2 K\big(\mu(s)\big)d\langle \mu_i, \mu_k \rangle (s)
		>
		B_{H}\big(A(T_*)\big),
	\end{equation}
	where
	\begin{equation*}
		B_{H}\big(A(T_*)\big) := -H\big(A(T_*)\big)+H\big(A(0)\big) + \sum_{\ell=1}^m \int_{0}^{T_*}D_{\ell}H\big(A(s)\big)dA_{\ell}(s).
	\end{equation*}
	When we apply the pathwise It\^o formula of Theorem~\ref{Thm : Ito formula} to the function $H\big(\langle\langle\mu\rangle\rangle(t)\big), ~~~ 0 \leq t \leq T$, the right-hand side of the above expression vanishes. Hence, the requirement \eqref{Ineq: separated case} becomes 
	\begin{equation*}
		K\big(\mu(T_*)\big)-\frac{1}{2}\sum_{i, k=1}^{d}\int_{0}^{T_*}\partial_{i, k}^2 K\big(\mu(s)\big)d\langle \mu_i, \mu_k \rangle (s) > K\big(\mu(0)\big)
	\end{equation*}
	and we are in very similar situation as in Theorem 5.1 of \cite{Karatzas:Ruf:2017}.
	
	\smallskip
	
	To be more precise, if $K$ takes non-negative values and is a `Lyapunov function' for the vector $\mu$ of market weights, in the sense that $\Gamma^{K}(t):=-\frac{1}{2}\sum_{i, k=1}^{d}\int_{0}^{t}\partial_{i, k}^2 K\big(\mu(s)\big)d\langle \mu_i, \mu_k \rangle (s)$ is nondecreasing, then the requirement $\Gamma^{K}(T_*)>K\big(\mu(0)\big)$ ensures strong relative arbitrage over every time-horizon $[0, t]$ with $T_* \leq t \leq T$ as in Theorem 5.1 of \cite{Karatzas:Ruf:2017}. Thus, in this `separated' case, we cannot achieve more than the result in Theorem 5.1 of \cite{Karatzas:Ruf:2017}, as all terms on the right-hand side of \eqref{Ineq: separated case} that involve $H$ vanish. This is because when we generate additively the trading strategy $\varphi$ in \eqref{Def: varphi} from a regular function $G$, only the partial derivatives of $G$ with respect to the market weights in \eqref{Def: vartheta2} are involved in $\varphi$, and this makes the $H$ term in \eqref{Eq: seperated case} meaningless in generating $\varphi$. Therefore, to be able to find new sufficient conditions for strong relative arbitrage, we need forms of $G$ more sophisticated than \eqref{Eq: seperated case}. All the examples of $G$ we develop in this paper from now onwards, are of those more elaborate forms.
\end{rem}

\medskip

From \eqref{Eq: value of ATS}, the value $V^{\varphi}$ at time $t$ of the additively generated trading strategy with respect to the market, has two additive components, $G\big(\mu(t), A(t)\big)$ and $\Gamma^{G}(t)$. In Theorem~\ref{thm: AGRA}, we derived the strong arbitrage condition from the ``nondecreasing property" of $\Gamma^{G}(\cdot)$, but there is no reason to differentiate between $G\big(\mu(t), A(t)\big)$ and $\Gamma^{G}(t)$. If the mapping $t \mapsto G\big(\mu(t), A(t)\big)$ is nondecreasing, it is possible derive a strong arbitrage condition like Theorem~\ref{thm: AGRA}, switching the role of $G\big(\mu(t), A(t)\big)$ and $\Gamma^{G}(t)$. However, it is difficult to find functions $G\big(\mu(t), A(t)\big)$ which are monotone in $t$, because $G$ must depend on the market weights $\mu(\cdot)$ and these fluctuate all the time. Thus, we have to `extract a nondecreasing structure' from the generating function $G\big(\mu(\cdot), A(\cdot)\big)$, and use this nondecreasing structure instead of $G$ to derive a new strong arbitrage condition. This is done as follows.

\medskip

\begin{thm} [Additively generated strong relative arbitrage when $\Gamma^G$ admits a lower bound]
	\label{thm: AGRA2}
	Fix a regular function $G: \mathbb{R}^d \times \mathbb{R}^m \rightarrow [0, \infty)$ for the pair $(\mu, A)$, where $\mu$ is the vector of market weights and $A$ is an $m$-dimensional function in $CBV([0, T], \mathbb{R}^m)$, such that the following conditions are satisfied:
	\begin{enumerate}[(i)]
		\item $V^{\varphi}(\cdot) = G\big(\mu(\cdot), A(\cdot)\big)+\Gamma^{G}(\cdot) \geq 0$, where the process $\Gamma^G(\cdot)$ is from \eqref{Def: gamma} or \eqref{Eq: gamma};
		\item there exists a function $F\big(\mu(t), A(t)\big)$ satisfying $G\big(\mu(t), A(t)\big) \geq F\big(\mu(t), A(t)\big)$ for all $t \in [0, T]$ and the mapping $t \mapsto F\big(\mu(t), A(t)\big)$ is nondecreasing;
		\item $\Gamma^{G}(\cdot) \geq -\kappa$ holds by some constant $\kappa$.
	\end{enumerate}
	For some real number $T_*>0$, suppose that
	\begin{equation}	\label{Con: additive arb2}
		F\big(\mu(T_*), A(T_*)\big) > G\big(\mu(0), A(0)\big) + \kappa
	\end{equation}
	holds. Then the additively generated strategy $\varphi$ of Definition~\ref{def: AG} is strong arbitrage relative to the market over every time horizon $[0, t]$ with $T_* \leq t \leq T$.
\end{thm}

\medskip

\begin{proof}
	The inequality \eqref{Con: strong arb2} is satisfied by the first condition above. From the last two conditions with \eqref{Eq: value of ATS} and \eqref{Con: additive arb2}, we obtain also the inequality \eqref{Con: strong arb}, since $V^{\varphi}(t)=G\big(\mu(t), A(t)\big)+\Gamma^G(t) \geq F\big(\mu(t), A(t)\big) - \kappa \geq F\big(\mu(T_*), A(T_*)\big) - \kappa > G\big(\mu(0), A(0)\big)=V^{\varphi}(0)$, for every $t \in [T_*, T]$.
\end{proof}

\medskip

In Theorem~\ref{thm: AGRA2}, the function $F\big(\mu(\cdot), A(\cdot)\big)$ can be seen as the `extracted nondecreasing structure' of $G$. This result states that the generating function $G$ can lead to strong arbitrage relative to the market without necessarily being ``Lyapunov'', as in Theorem~5.1 of \cite{Karatzas:Ruf:2017}. There can be strong relative arbitrage even if $\Gamma^{G}(\cdot)$ is nonincreasing. This is intuitively plausible already on the basis of the representation \eqref{Eq: value of ATS} when $G\big(\mu(\cdot), A(\cdot)\big)$ grows faster than $\Gamma^{G}(\cdot)$ decays. Some applications of Theorem~\ref{thm: AGRA2} will appear in Section~\ref{sec: 6} (Example~\ref{Ex: entropy with MIN} and Example~\ref{Ex: entropy with LOGLOG}).

\bigskip

\subsection{Multiplicatively generated strong relative arbitrage}

In this subsection, in order to simplify the arguments, we assume that the regular function $G$ takes only nonnegative values and satisfies $G\big(\mu(0), A(0)\big)=1$. This normalization can be achieved by replacing $G$ by $G/G\big(\mu(0), A(0)\big)$ if $G\big(\mu(0), A(0)\big)>0$, or by $G+1$ if $G\big(\mu(0), A(0)\big)=0$. Though we shall not use in later sections the following result, which comes from Theorem 5.2 of \cite{Karatzas:Ruf:2017}, we state here for completeness.

\medskip

\begin{thm} [Multiplicatively generated strong relative arbitrage]
	\label{thm: MGRA}
	Let us fix a regular function $G: \mathbb{R}^d \times \mathbb{R}^m \rightarrow [0, \infty)$ for the pair $(\mu, A)$ with the market weights $\mu$ and some $m$-dimensional function $A \in CBV([0, T], \mathbb{R}^m)$.	For some real number $T_*>0$, suppose that there exists an $\epsilon=\epsilon(T_*)>0$ satisfying
	\begin{equation}	\label{Con: multiplicative arb}
		\Gamma^G(T_*)>1+\epsilon.
	\end{equation}
	Then, there exists a constant $c=c(T_*, \epsilon)>0$ such that the trading strategy $\psi^{(c)}=(\psi_1^{(c)}, \cdots, \psi_d^{(c)})'$, multiplicatively generated by the regular function 
	\begin{equation*}
		G^{(c)}:=\frac{G+c}{1+c}	
	\end{equation*}
	as in Definition~\ref{def: MG}, is strong arbitrage relative to the market over the time-horizon $[0, T_*]$; as well as over every time-horizon $[0, t]$ with $T_* \leq t \leq T$, if $t \mapsto \Gamma^G(t)$ is nondecreasing.
\end{thm}

\medskip

In Theorem~\ref{thm: MGRA}, we had to construct the ``shifted'' function $G^{(c)}$, which was useful but also rather extraneous. However, for functions $G$ which are bounded from below and above, we have the following novel condition leading to multiplicatively generated strong arbitrage.

\medskip

\begin{thm} [Multiplicatively generated strong relative arbitrage when $\Gamma^G$ is nondecreasing]
	\label{thm: MGRA2}
	Let us fix a function $G: \mathbb{R}^d \times \mathbb{R}^m \rightarrow (0, \infty)$ which is regular for the pair $(\mu, A)$ with the market weights $\mu$ and some $m$-dimensional function $A \in CBV([0, T], \mathbb{R}^m)$, and satisfies the following conditions:
	\begin{enumerate}[(i)]
		\item $G$ is bounded away from zero and infinity, i.e., there exist positive constants $\alpha, \beta$ such that $0 < \alpha \leq G\big(\mu(\cdot), A(\cdot)\big) \leq \beta$;
		\item $\Gamma^G$ is nondecreasing.
	\end{enumerate}
	For some real number $T_*>0$, suppose that
	\begin{equation}	\label{Con: multiplicative arb2}
		\Gamma^G(T_*) > \beta \log \Big(\frac{1}{\alpha}\Big)
	\end{equation}
	holds. Then, the multiplicatively generated strategy $\psi$ of Definition~\ref{def: MG} is strong arbitrage relative to the market over every time-horizon $[0, t]$ with $T_* \leq t \leq T$.
\end{thm}

\medskip

\begin{proof}
	First, we note that $V^{\psi}(\cdot) > 0$ from \eqref{Eq: value of MTS}. Then, we take the logarithm to the both sides of \eqref{Eq: value of MTS} to obtain
	\begin{align*}
		\log V^{\psi}(t) 
		&= \log G\big(\mu(t), A(t)\big) + \int_0^t \frac{d\Gamma^G(s)}{G\big(\mu(s), A(s)\big)}
		\\
		& \geq \log \alpha + \frac{1}{\beta}\Gamma^G(t)
		\geq \log \alpha + \frac{1}{\beta}\Gamma^G(T_*)
		\\
		& > 0 = \log G\big(\mu(0), A(0)\big) = \log V^{\psi}(0),
	\end{align*}
	for all $T_* \leq t \leq T$, due to the conditions \textit{(i), (ii)} and \eqref{Con: multiplicative arb2} above, and the result follows. Here $G\big(\mu(0), A(0)\big) = 1$, because of the normalization on $G$ imposed at the beginning of this subsection.
\end{proof}

\medskip

\begin{rem}
	Since the market weights $\mu_i$, $i=1, \cdots, d$ and the continuous function $A$ are bounded on the compact interval $[0, T]$, a regular function $G$ depending on the pair $(\mu, A)$ is also bounded. Thus, the condition \textit{(i)} in Theorem~\ref{thm: MGRA2} just requires that the lower bound $\alpha$ should be strictly bigger than $0$. Also, in \eqref{Con: multiplicative arb2}, finding tighter bounds $\alpha$, $\beta$ of $G$ yields smaller $T_*$ satisfying the arbitrage condition \eqref{Con: multiplicative arb2}. See Remark~\ref{Re: Reducing threshold in entropy example} for further discussion regarding the bounds on $G$ in the case of specific entropy function.
\end{rem}

\medskip

The conditions of Theorem~\ref{thm: MGRA2} resemble those of Theorem~\ref{thm: AGRA}. We also have the following formulation, which is similar to Theorem~\ref{thm: AGRA2}.
 
\medskip

\begin{thm} [Multiplicatively generated strong relative arbitrage when $\Gamma^G$ is nonincreasing]
	\label{thm: MGRA3}
	Fix a regular function $G: \mathbb{R}^d \times \mathbb{R}^m \rightarrow (0, \infty)$ for the pair $(\mu, A)$, where $\mu$ is the vector of market weights and $A$ an $m$-dimensional function in $CBV([0, T], \mathbb{R}^m)$, such that the following conditions hold:
	\begin{enumerate}[(i)]
		\item there exists a function $F\big(\mu(t), A(t)\big)$ satisfying $G\big(\mu(t), A(t)\big) \geq F\big(\mu(t), A(t)\big) > 0$ for all $t \in [0, T]$, and the mapping $t \mapsto F\big(\mu(t), A(t)\big)$ is nondecreasing;
		\item $\Gamma^G(\cdot)$ is nonincreasing and $\Gamma^{G}(\cdot) \geq -\kappa$ holds by some positive constant $\kappa$.
	\end{enumerate}
	For some real number $T_*>0$, suppose that
	\begin{equation}	\label{Con: multiplicative arb3}
		\log F\big(\mu(T_*), A(T_*)\big) > \frac{\kappa}{F\big(\mu(0), A(0)\big)}
	\end{equation}
	holds. Then the multiplicatively generated strategy $\psi$ of Definition~\ref{def: MG} is strong arbitrage relative to the market over every time horizon $[0, t]$ with $T_* \leq t \leq T$.
\end{thm}

\begin{proof}
	First, note that $\Gamma^G(\cdot)$ is nonpositive, because of the condition \textit{(ii)} and the fact $\Gamma^G(0)=0$.
	Again, from \eqref{Eq: value of MTS}, we have
	\begin{align*}
		\log V^{\psi}(t) 
		&= \log G\big(\mu(t), A(t)\big) + \int_0^t \frac{d\Gamma^G(s)}{G\big(\mu(s), A(s)\big)}
		\geq \log F\big(\mu(t), A(t)\big) - \max_{0 \leq s \leq t} \Big(\frac{\kappa}{G\big(\mu(s), A(s)\big)}\Big)
		\\
		&= \log F\big(\mu(t), A(t)\big) - \frac{\kappa}{\min_{0 \leq s \leq t} G\big(\mu(s), A(s)\big)}
		\geq \log F\big(\mu(t), A(t)\big) - \frac{\kappa}{\min_{0 \leq s \leq t} F\big(\mu(s), A(s)\big)}
		\\
		& \geq \log F\big(\mu(T_*), A(T_*)\big) - \frac{\kappa}{F\big(\mu(0), A(0)\big)}
		> 0 = \log G\big(\mu(0), A(0)\big) = \log V^{\psi}(0),
	\end{align*}
	for all $T_* \leq t \leq T$, by virtue of the conditions \textit{(i), (ii)} and \eqref{Con: multiplicative arb3}.
\end{proof}

\medskip

The following example provides a condition for strong relative arbitrage more general than Example 5.5 of \cite{Karatzas:Ruf:2017}, by involving an additional function $A$ into the generating function $G$. We specifically use $A = \mu^* = (\mu^*_1, \cdots, \mu^*_d)$, the vector consisting of the running maxima of the market weights 
\begin{equation*}
	A_i(t) \equiv \mu^*_i(t) := \max_{0 \leq s \leq t}\mu_i(s), \qquad i=1, \cdots, d.
\end{equation*}

\bigskip

\begin{example} [Quadratic function]
	For fixed constant $c \in \mathbb{R}$ and $p>0$, consider the following function
	\begin{align*}
		G^{(c, p)}\big(\mu(t), \mu^*(t)\big) &:= c - \sum_{i=1}^{d} \big(\mu_i(t)\big)^2 - p \sum_{i=1}^{d} \mu_i(t)\mu^*_i(t)
		\\
		& \ =  c - \sum_{i=1}^{d} \big(\mu_i(t)\big)^2 - p \sum_{i=1}^{d} \mu_i(t)\big\{\max_{0 \leq s \leq t}\mu_i(s)\big\}.
	\end{align*}
	This is the same as $Q^{(c)}$ in Example 5.5 of \cite{Karatzas:Ruf:2017} except for the last term. Note that $G^{(c, p)}$ takes values in the interval $\big[c-(1+p), \ c-\frac{1}{d}(1+p)\big]$. After some straightforward computation of partial derivatives, we have
	\begin{align*}
		D_iG^{(c, p)}\big(\mu(t), \mu^*(t)\big) &= -p\mu_i(t),
		\\
		\partial_iG^{(c, p)}\big(\mu(t), \mu^*(t)\big) &= -2\mu_i(t)-p\mu^*_i(t),
		\\
		\partial^2_{i, i}G^{(c, p)}\big(\mu(t), \mu^*(t)\big) &= -2,
	\end{align*}
	for $i=1, \cdots, d$, and using these expressions along with \eqref{Eq: gamma}, we obtain
	\begin{equation*}
		\Gamma^{G^{(c, p)}}(t) = \sum_{i=1}^{d}\int_{0}^{t}p\mu_i(s)d\mu^*_i(s) + \sum_{i=1}^{d}\langle \mu_i \rangle(t).
	\end{equation*}
	As $\mu^*_i(\cdot)$ is nondecreasing and $p\mu_i(\cdot) \geq 0$, the integral term is always non-negative and nondecreasing in $t$, which makes $\Gamma^{G^{(c, p)}}(\cdot)$ nondecreasing and non-negative. Also, using the property that the nondecreasing process $\mu^*_i(\cdot)$ is flat off the set $\{s \geq 0 : \mu_i(s)=\mu^*_i(s)\}$, we have
	\begin{equation*}
		\int_{0}^{t}\mu_i(s)d\mu^*_i(s) = \int_{0}^{t}\mu^*_i(s)d\mu^*_i(s) = \frac{1}{2}\big\{\big(\mu^*_i(t)\big)^2 - \big(\mu^*_i(0)\big)^2\big\},
	\end{equation*}
	thus also
	\begin{equation*}
		\Gamma^{G^{(c, p)}}(t) = \frac{p}{2}\sum_{i=1}^{d}\big\{\big(\mu^*_i(t)\big)^2 - \big(\mu^*_i(0)\big)^2\big\} + \sum_{i=1}^{d}\langle \mu_i \rangle(t).
	\end{equation*}
	Since $G^{(1+p, p)} \geq 0$, let us consider the case $c=1+p$ from now on.	Using the same argument as in the proof of Theorem~\ref{thm: AGRA}, the condition
	\begin{equation} 	\label{Ex: running maxima arb con}
		\frac{p}{2}\sum_{i=1}^{d}\big\{\big(\mu^*_i(T)\big)^2 - \big(\mu^*_i(0)\big)^2\big\} + \sum_{i=1}^{d}\langle \mu_i \rangle(T) > G^{(1+p, p)}\big(\mu(0), \mu^*(0)\big),
	\end{equation}
	where
	\begin{equation*}
		G^{(1+p, p)}\big(\mu(0), \mu^*(0)\big) = (1+p)\Big\{ 1-\sum_{i=1}^{d} \big(\mu_i(0)\big)^2 \Big\} > 0,
	\end{equation*}
	yields a strategy which is strong relative arbitrage with respect to the market on $[0, T]$. If we compare the condition \eqref{Ex: running maxima arb con} with the condition
	\begin{equation}	\label{Ex: running maxima arb con2}
		\sum_{i=1}^{d}\langle \mu_i \rangle(T) > 1-\sum_{i=1}^{d} \big(\mu_i(0)\big)^2,
	\end{equation}
	that is, (5.4) of Example 5.5 in \cite{Karatzas:Ruf:2017}, there is a trade-off between the left- and the right-hand sides. The presence of the extra nondecreasing term $(p/2)\sum_{i=1}^{d}\big\{\big(\mu^*_i(T)\big)^2 - \big(\mu^*_i(0)\big)^2\big\}$ in \eqref{Ex: running maxima arb con}, guarantees that its left-hand side grows faster than the left-hand side of \eqref{Ex: running maxima arb con2}, as $T$ increases; but we also have a bigger constant on the right-hand side of \eqref{Ex: running maxima arb con}, namely,
	\begin{equation*}
		(1+p)\Big\{ 1-\sum_{i=1}^{d} \big(\mu_i(0)\big)^2 \Big\}
		>
		1-\sum_{i=1}^{d} \big(\mu_i(0)\big)^2.
	\end{equation*}
	Thus, by choosing the value of $p$ wisely, we can obtain bounds for the times $T$ for which there is strong relative arbitrage with respect to the market over $[0, T]$, better than those of Example 5.5 in \cite{Karatzas:Ruf:2017}.
\end{example}

\medskip

More interesting applications of Theorems~\ref{thm: AGRA}, \ref{thm: AGRA2}, \ref{thm: MGRA2}, and \ref{thm: MGRA3} appear in Section~\ref{sec: 6}.

\bigskip

\bigskip

\bigskip

\section{Tanaka's formula for constructing trading strategies}
 \label{sec: 5}

In the previous sections, we used the pathwise It\^o formula (Theorem~\ref{Thm : Ito formula}), instead of the usual It\^o formula for semimartingales, to construct trading strategies in a pathwise manner. In this section, we generalize Stochastic Portfolio Theory~(SPT) in a different direction: we apply the pathwise Tanaka formula (Generalized It\^o's formula) with appropriately defined local time, in building up trading strategies. The It{\^o} formula requires the existence of a second derivative, whereas the Tanaka formula is applicable to `weakly differentiable' functions; this broadens the class of functions from which we generate trading strategies. First, we develop some definitions and notation, and introduce the pathwise Tanaka formula. Then, we construct trading strategies generated additively and multiplicatively, and in a manner similar to that of Section~\ref{sec: 3}, but from generating functions less smooth than those used there. Finally, relevant strong relative arbitrage conditions and some examples follow.

\smallskip

\subsection{Pathwise local time and Tanaka formula}
We fix a refining sequence $(\mathbb{T}_n)_{n \in \mathbb{N}}$ of partitions of the interval $[0, T]$, whose mesh size goes to zero as $n \rightarrow 0$, as in the introduction of Section~\ref{sec: 2}. We also consider an $\mathbb{R}$-valued continuous function $X$ defined on the compact interval $[0, T]$, thought of here as representing a value of an individual asset which fluctuates over time. With these ingredients, we present the measure-theoretic notion of quadratic variation of $X$ along $\mathbb{T}=(\mathbb{T}_n)_{n \in \mathbb{N}}$.

\medskip

\begin{defn}	\label{Def: QV}
	A continuous function $X\in C([0, T], \mathbb{R})$ is said to have finite \textit{quadratic variation along a given sequence of partitions $\mathbb{T}=(\mathbb{T}_n)_{n \in \mathbb{N}}$} of $[0,T]$, if the mesh size
	\begin{equation}		\label{def: mesh}
		||\mathbb{T}_n|| := \max_{[t_j, t_{j+1}] \in \mathbb{T}_n} |t_{j+1}-t_j|,
	\end{equation}
	goes to zero and the sequence of measures 
	\begin{equation*}
		\mu^n := \sum_{[t_j, t_{j+1}] \in \pi_n} \big|X(t_{j+1}) - X(t_j)\big|^2 \cdot \delta_{t_j}
	\end{equation*}
	converges vaguely to a locally finite measure $\mu$ without atoms as $n \rightarrow \infty$, where $\delta_t$ denotes the Dirac measure at $t \in [0, T]$. We write $Q(\mathbb{T})$ for the collection of all continuous functions having quadratic variation along $\mathbb{T}$. We call $\langle X\rangle(t) :=\mu([0, t])$ for $t \in [0, T]$, the quadratic variation of $X$.
\end{defn}

\medskip

For a sequence of measures $(\mu^n)_{n \in \mathbb{N}}$ on $[0, T]$, vague convergence is equivalent to the pointwise convergence of their cumulative distribution functions at all continuity points of the limiting function. If the limiting distribution function is continuous, the convergence is uniform. Thus we are led to the following result.

\medskip

\begin{lem}
	\label{Lem: pointwise convergence}
	Let $X$ be a function in $C([0, T], \mathbb{R})$. The function $X$ belongs to $Q(\mathbb{T})$ if, and only if, there exists a continuous function $\langle X\rangle$ such that for every $t \in [0, T]$,
	\begin{equation} \label{Eq: pointwise convergence}
		\sum_{\substack{[t_j, t_{j+1}] \in \mathbb{T}_n \\ t_j \leq t}} | X(t_{j+1}) - X(t_j)|^2 \xrightarrow{n \rightarrow \infty} \langle X\rangle(t).
	\end{equation}
	If this property holds, the convergence in \eqref{Eq: pointwise convergence} is uniform.
\end{lem}

\medskip

From this Lemma, the quadratic variation $\langle X \rangle$ of $X$ in Definition~\ref{Def: QV} coincides with that of $X$ in Definition~\ref{Def: quadratic cov}. However, there is a notion of quadratic `covariation' $\langle X_i, X_j \rangle (\cdot)$ between different components $X_i$, $X_j$ of a $d$-dimensional vector $X$ in \eqref{Def: quadratic}, whereas $\langle X\rangle(\cdot)$ in Definition~\ref{Def: QV} is defined in terms of the individual function $X$.

\bigskip

\begin{rem}
	The assumption in Definition~\ref{Def: QV} that the mesh size in \eqref{def: mesh} goes to zero as $n \rightarrow \infty$, imposed on the sequence $(\mathbb{T}_n)_{n \in \mathbb{N}}$ of partitions, is actually stronger than the usual assumption on the sequence of partitions in other works involving the pathwise local time. For example, in \cite{PerkowskiPromel2}, \cite{Davis2018}, \cite{Cont_Perkowski}, the authors define the `oscillation' of the function $X$ along the partition $\mathbb{T}_n$ as
	\begin{equation}	\label{Def: oscillation}
		osc(X, \mathbb{T}_n) := \max_{[t_j, t_{j+1}] \in \mathbb{T}_n} \max_{r, s \in [t_j, t_{j+1}]} |X(s) - X(r)|,
	\end{equation}
	and require $osc(X, \mathbb{T}_n) \rightarrow 0$ as $n \rightarrow \infty$ instead of the mesh size going to zero. This is because it is enough to work with Lebesgue partitions generated by $X$ when defining the pathwise local time and deriving the pathwise Tanaka formula. Since the function $X$ is uniformly continuous on the compact interval $[0, T]$, the decrease to zero of the mesh size does imply that the oscillation of $X$ also shrinks to zero.
	
	\smallskip
	
	One reason for the stronger condition on $(\mathbb{T}_n)_{n \in \mathbb{N}}$ used here, is to follow our original definition of pathwise quadratic covariation/variation in Definition~\ref{Def: quadratic cov}. Another reason is that we are going to involve an additional (vector of) continuous function $A$ when generating trading strategies, and the oscillation of $A$ also has to shrink to zero along the sequence of partitions $(\mathbb{T}_n)_{n \in \mathbb{N}}$. In other words, by using the `mesh' assumption instead of the `oscillation', we can get rid of such `dependence' of the sequence of partitions $(\mathbb{T}_n)_{n \in \mathbb{N}}$ on both $X$ and $A$.
\end{rem}

\bigskip

The very first definition of the pathwise local time was introduced in the unpublished diploma thesis of \cite{Wuermli}. This original local time is called ``$\mathbb{L}^2$-local time'' of a path $X$ along a sequence of partitions $\mathbb{T} = (\mathbb{T}_n)_{n\in \mathbb{N}}$. Using this notion of local time, Wuermli showed the following equation \eqref{Eq : tanaka formula} for $f \in H^2(\mathbb{R}, \mathbb{R})$, where $H^2(\mathbb{R}, \mathbb{R})$ is the Sobolev space of functions two times weakly differentiable in $\mathbb{L}^2(\mathbb{R}, \mathbb{R})$. Since then, many versions of pathwise Tanaka formulas~(generalized It\^o formulas) and different definitions of local times have been introduced and studied; these vary according to the regularity of the path $X$, the function $f$, and the notion of ``convergence for local time''. Weaker convergence in defining a local time requires more regularity on the part of the function $f$, for the Tanaka formula \eqref{Eq : tanaka formula} to hold. Some of these versions are stated in Section~2 of \cite{PerkowskiPromel2} for continuous paths with quadratic variation. Similar results for rougher paths (with finite $p$-th variation, $p>2$) can be found in Section~3 of \cite{Cont_Perkowski}. Among these, we present here the following version of local time and Tanaka's formula, which we consider most appropriate in our setting.

\medskip

With the notation
\begin{equation}		\label{Def: parenthesis}
	\llparenthesis a, b \rrbracket = \begin{cases}
	(a, b], \qquad a \leq b,
	\\
	(b, a], \qquad b \leq a,
\end{cases}
\end{equation}

\medskip

we have the following definition of continuous local time.

\medskip

\begin{defn} [Continuous local time]
	\label{Def : local time}
	We say that the continuous function $X : [0, T] \rightarrow \mathbb{R}$ has \textit{a continuous local time along the given nested sequence of partitions $\mathbb{T} = (\mathbb{T}_n)_{n\in \mathbb{N}}$}, if $\lim_{n \rightarrow \infty} ||\mathbb{T}_n|| = 0$, the `discrete local times'
	\begin{equation}		\label{Def: discrete local time}
		x \mapsto L^{X, \mathbb{T}_n}_t(x) := \sum_{\substack{[t_j, t_{j+1}] \in \mathbb{T}_n \\ t_j \leq t}} \mathbbm{1}_{\llparenthesis X_{t_j}, X_{t_{j+1}} \rrbracket} (x) \big| X_{t_{j+1}} - x \big|
	\end{equation}
	converge uniformly to a continuous limit $x \mapsto L_t^{X, \mathbb{T}}(x)$ as $n \rightarrow \infty$ for every fixed $t \in [0, T]$, and the resulting mapping $(t, x) \mapsto L_t^{X, \mathbb{T}}(x)$ is jointly continuous. We call this limit \textit{continuous local time of $X$ along $\mathbb{T}$}, and write $\pazocal{L}^c(\mathbb{T})$ for the collection of all functions $X$ in $C([0, T], \mathbb{R})$ having a continuous local time along the given nested sequence of partitions $\mathbb{T} = (\mathbb{T}_n)_{n \in \mathbb{N}}$.
\end{defn}

\medskip

The existence of continuous local time for typical price paths is shown in Theorem~3.5 of \cite{PerkowskiPromel2}. In order to simplify notation, we shall write $L_t^X(x)$ or simply $L_t(x)$, whenever the context is unambiguous. With the definition of continuous local time, we state the following version of the pathwise Tanaka formula. The proof is given in the Appendix.

\medskip

\begin{thm} [Pathwise Tanaka formula for paths with finite quadratic variation]
	\label{thm: tanaka formula}
	Let $X \in \pazocal{L}^c(\mathbb{T})$ and $f : \mathbb{R} \rightarrow \mathbb{R}$ be absolutely continuous with right-continuous Radon-Nikod{\'y}m derivative $f'$ of finite variation.	Then, the one-dimensional F{\"o}llmer-It\^o integral
	\begin{equation}	\label{Eq: compensated Riemann sums}
		\int_0^t f'\big(X(s)\big)dX(s) := \lim_{n \rightarrow \infty} \sum_{\substack{[t_j, t_{j+1}] \in \mathbb{T}_n \\ t_j \leq t}} f'\big(X(t_j)\big)\big(X(t_{j+1})-X(t_j)\big)
	\end{equation}
	exists, and we have the generalized change of variable formula
	\begin{equation} \label{Eq : tanaka formula}
		f\big(X(t)\big) - f\big(X(0)\big) = \int_0^t f'\big(X(s)\big)dX(s) + \int_{\mathbb{R}} L_t(x) df'(x), \qquad 0 \leq t \leq T.
	\end{equation}	
\end{thm}

\bigskip

If $f$ belongs to the space $C^2(\mathbb{R}, \mathbb{R})$, we obtain from this
\begin{equation*}
	\frac{1}{2}\int_0^t f''\big(X(s)\big)d\langle X \rangle(s) = \int_{\mathbb{R}} L_t(x)f''(x) dx
\end{equation*}
by comparing the last terms of \eqref{Eq : Ito formula} and \eqref{Eq : tanaka formula}. Furthermore, by setting $g(\cdot) = f''(\cdot)$ for any continuous function $g$ and by the fact that the indicator function $\mathbbm{1}_A(\cdot)$ for any Borel set $A \in \mathcal{B}(\mathbb{R})$ can be approximated by continuous functions, we also have the ``occupation density formula''
\begin{equation*}
	\frac{1}{2}\int_0^t \mathbbm{1}_A\big( X(s) \big) d[X](s)
	= \int_{A} L_t(x) dx.
\end{equation*}

\smallskip

We state now pathwise versions of classical Tanaka-Meyer formulas as a corollary to Theorem~\ref{thm: tanaka formula}.

\medskip

\begin{cor}	\label{Cor : Tanaka-Meyer}
	For a function $X \in \pazocal{L}^{c}(\mathbb{T})$, the pathwise Tanaka-Meyer formulas
	\begin{equation}	\label{Eq : tanaka-meyer1}
		L_t(a) = (X_t-a)^+ - (X_0-a)^+ - \int_0^t \mathbbm{1}_{(a, \infty)}(X_s) dX_s,
	\end{equation}
	\begin{equation}	\label{Eq : tanaka-meyer2}
		L_t(a) = (X_t-a)^- - (X_0-a)^- + \int_0^t \mathbbm{1}_{(-\infty, a)}(X_s) dX_s
	\end{equation}
	and
	\begin{equation}	\label{Eq : tanaka-meyer3}
		2L_t(a) = |X_t-a| - |X_0-a| - \int_0^t {\rm{sign}} (X_s-a) dX_s
	\end{equation}
	hold for all $(t, a) \in [0, T] \times \mathbb{R}$ with the notation ${\rm{sign}}(x) = \begin{cases} ~~~1, ~~~ \text{if}~ x > 0, \\ -1, ~~~ \text{if}~ x \leq 0. \end{cases}$ Here, the integral terms represent pointwise limits as \eqref{Eq: compensated Riemann sums}.
\end{cor}

\bigskip

\subsection{Construction of additively generated trading strategies}
Now we recall the pair $(\mu, A)$ in Section~\ref{sec: 3}, where the vector $\mu = (\mu_1, \cdots \mu_d)'$ represents market weights defined in \eqref{Def: market weights}, and $A = (A_1, \cdots, A_d)$ is an auxiliary function in $CBV([0, T], \mathbb{R}^d)$. We assume that $\mu$ and $A$ have the same dimension $d$. Also, in this section, we assume that each component $\mu_i$ is a continuous function with finite quadratic variation $\langle \mu_i \rangle$ in the sense of Definition~\ref{Def: QV}, and belonging to $\pazocal{L}^c(\mathbb{T})$, i.e., admitting a continuous local time, for every $i = 1, \cdots, d$. Then, we set 
\begin{equation}	\label{def: X}
	X_i := \mu_i - A_i, \qquad i = 1, \cdots, d,
\end{equation}
and assume that each $X_i$ is also in $\pazocal{L}^c(\mathbb{T})$. For any functions $f_i$, $i=1, \cdots, d,$ satisfying the conditions in Theorem~\ref{thm: tanaka formula}, we define the generating function $G$ for the pair $(\mu, A)$ as
\begin{equation}	\label{def: generating function}
	G\big(\mu(t), A(t)\big) := \sum_{i=1}^d f_i(X_i(t)) = \sum_{i=1}^d f_i\big(\mu_i(t)-A_i(t)\big), \qquad 0 \leq t \leq T.
\end{equation}
We can only consider such generating function $G$ of the form \eqref{def: generating function}, because there is no `multidimensional Tanaka formula' that can be applied to $G$ directly. However, we can apply instead Theorem~\ref{thm: tanaka formula} to each components $f_i(X_i(t))$, and sum up to obtain
\begin{align}
	G\big(\mu(t), A(t)\big) &= \sum_{i=1}^d \bigg\{ f_i(X_i(0)) + \int_0^t f'_i\big(X_i(s)\big)dX_i(s) + \int_{\mathbb{R}} L_t^{X_i}(x)df'_i(x) \bigg\}		\label{Eq: decomp of G}
	\\
	&=G\big(\mu(0), A(0)\big) + \sum_{i=1}^d \int_0^t \partial_iG\big(\mu(s), A(s)\big)dX_i(s) + \sum_{i=1}^d\int_{\mathbb{R}} L_t^{(\mu_i-A_i)}(x)df'_i(x),		\nonumber
\end{align}
where we denote
\begin{equation}	\label{def: vartheta}
	\vartheta_i(t) := \partial_i G\big(\mu(t), A(t)\big) := f'_i\big(X_i(t)\big), \qquad i = 1, \cdots, d, \qquad 0 \leq t \leq T,
\end{equation}
as in \eqref{Def: vartheta2}. Here, we recall from Theorem~\ref{thm: tanaka formula} that each $f'_i$ is the RCLL derivative of $f_i$, and a function of bounded variation. 
Furthermore, the F{\"o}llmer-It\^o integral in \eqref{Eq: decomp of G}, defined via the recipe \eqref{Eq: compensated Riemann sums}, can be decomposed as
\begin{align}
	\int_0^t f'_i\big(X_i(s)\big)dX_i(s) 
	&= \lim_{n \rightarrow \infty} \sum_{\substack{[t_j, t_{j+1}] \in \mathbb{T}_n \\ t_j \leq t}} f'_i\big(X_i(t_j)\big)\big(X_i(t_{j+1})-X_i(t_j)\big)			\label{Eq: follmer}
	\\
	&=\lim_{n \rightarrow \infty} \sum_{\substack{[t_j, t_{j+1}] \in \mathbb{T}_n \\ t_j \leq t}} f'_i\big(X_i(t_j)\big)\big(\mu_i(t_{j+1})-\mu_i(t_j)\big)		\label{Eq: follmer_mu}
	\\
	&~~~-\lim_{n \rightarrow \infty} \sum_{\substack{[t_j, t_{j+1}] \in \mathbb{T}_n \\ t_j \leq t}} f'_i\big(X_i(t_j)\big)\big(A_i(t_{j+1})-A_i(t_j)\big),	\label{Eq: follmer_A}
\end{align}
because the last limit \eqref{Eq: follmer_A} exists as $A \in CBV([0, T], \mathbb{R}^d)$. Thus, the limit \eqref{Eq: follmer_mu} also exists, and we denote the two limits \eqref{Eq: follmer_mu} and \eqref{Eq: follmer_A} as $\int_0^t f'_i\big(X_i(s)\big)d\mu_i(s)$, $\int_0^t f'_i\big(X_i(s)\big)dA_i(s)$, respectively. For the generating function $G$ in \eqref{def: generating function}, we define the Gamma function $\Gamma^G$ as in \eqref{Def: gamma}, namely
\begin{align}
	\Gamma^G(t) &:= G\big(\mu(0), A(0)\big) - G\big(\mu(t), A(t)\big) + \sum_{i=1}^d \int_0^t \vartheta_i(s) d\mu_i(s)	\nonumber
	\\
	&~=\sum_{i=1}^d \int_0^t \vartheta_i(s)dA_i(s) - \sum_{i=1}^d \int_{\mathbb{R}} L_t^{(\mu_i-A_i)}(x)df'(x), \qquad 0 \leq t \leq T.		\label{Def: Gamma}
\end{align}

\smallskip

The last equation is from \eqref{Eq: decomp of G}, and we note that $\Gamma^G(\cdot)$ is of bounded variation again. We proceed now in the manner of \eqref{Def: varphi}-\eqref{Def: defect of B}, to construct the additively generated trading strategy.

\medskip

\begin{defn} [Additive generation]		\label{Def : ag}
	We say that the trading strategy $\varphi = (\varphi_1, \cdots, \varphi_d)'$ defined as
	\begin{equation}	\label{def: varphi}
		\varphi_i(t) := \vartheta_i(t) - Q^{\vartheta}(t) - C(0), \quad \qquad i=1, \cdots, d, \qquad 0 \leq t \leq T,
	\end{equation}
	is \textit{additively generated} by the function $G$ in \eqref{def: generating function}. Here,
	\begin{equation}		\label{Def: dsf}
		Q^{\vartheta}(t):=\sum_{i=1}^d \big\{\vartheta_i(t)\mu_i(t)- \vartheta_i(0)\mu_i(0) \big\}-\int_{0}^{t} \sum_{i=1}^{d} \vartheta_i(s)d\mu_i(s)
	\end{equation}
	is the defect of self-financibility at time $t \in [0, T]$ as defined in \eqref{Def: defect of SF}, and
	\begin{equation*}
		C(0) := \sum_{i=1}^d \vartheta_i(0)\mu_i(0) - G\big(\mu(0), A(0)\big)
	\end{equation*}
	the defect of balance at $t=0$.
\end{defn}

\medskip

We also have the following result, which can be proved by analogy with Proposition~\ref{prop: additive generation}. Note that the Gamma function $\Gamma^G(\cdot)$ below takes the form of \eqref{Def: Gamma}, not of \eqref{Eq: gamma}.

\begin{prop}
	\label{prop: additive generation2}
	The trading strategy $\varphi$, generated additively as in \eqref{def: varphi} by the function $G$ of \eqref{def: generating function} for the pair $(\mu, A)$, has value
	\begin{equation}	\label{Eq: value of ATS2}
		V^{\varphi}(t):= \sum_{i=1}^d \varphi_i(t)\mu_i(t)=G\big(\mu(t), A(t)\big)+\Gamma^G(t), \quad 0 \leq t \leq T,
	\end{equation}
	with $\Gamma^G(\cdot)$ defined as in \eqref{Def: Gamma}, and its components can be represented in the equivalent form
	\begin{align}
	\varphi_i(t)&=\vartheta_i(t)+\Gamma^G(t)+G\big(\mu(t), A(t)\big)-\sum_{j=1}^{d}\vartheta_j(t)\mu_j(t) 	\label{Eq: varphi2}
	\\
	&=\vartheta_i(t)+V^{\varphi}(t)-\sum_{j=1}^{d}\vartheta_j(t)\mu_j(t), \qquad \text{for} \quad i=1, \cdots, d. \nonumber
	\end{align}
\end{prop}

\medskip

The sufficient conditions for strong relative arbitrage effected by additively generated trading strategies, presented in Section~\ref{sec : 4.1}, can be applied in a similar manner to the strategies $\varphi$ of Definition~\ref{Def : ag}.

\medskip

Concave functions such as $x \mapsto -x^2$ and $x \mapsto -x\log x$, when used to generate trading strategies, produce nondecreasing Gamma functions $\Gamma^G$ as in \eqref{Def: gamma}; this is because these functions have negative semidefinite Hessians $\partial^2 G = (\partial_{i, k}G)_{1 \leq i, k \leq d}$, which play the role of the integrand of the last integral in \eqref{Eq: gamma}. Such concavity is known to lead to ``diversity-weighted'' investment strategies, as explained in Definition~3.4.1 of \cite{Fe}. However, these concave functions had to be in $C^2$ to apply It\^o's rule. Now, we can use concave but not differentiable functions, while still being able to generate portfolio with the help of Tanaka formula. Typical examples are $x \mapsto -x^+ := -\max(x, 0)$ and $x \mapsto -x^- := -\min(x, 0)$.

\bigskip

\begin{example} [\textit{On the ``size effect"}]	\label{Ex: alpha}
	Consider a constant $\alpha \in (0, 1)$ and a function
	\begin{equation*}
		f(x) := \frac{1}{d}-(x-\alpha)^+,
	\end{equation*}
	where $x^+ := \max\{x, 0\}$ and $d$ is the dimension of the market weight vector $\mu$. Note that $f$ satisfies the conditions in Theorem~\ref{thm: tanaka formula}. Then, for the pair $(\mu, A)$ with $A \equiv 0$, we have $X \equiv \mu$ in \eqref{def: X}, and set $f = f_i$  for $i=1, \cdots, d$ to obtain the generating function in \eqref{def: generating function} as
	\begin{equation}	\label{Def: G in alpha}
		G(\mu(t)) = 1 - \sum_{i=1}^d \big(\mu_i(t)-\alpha\big)^+,
	\end{equation}
	which is nonnegative by construction. Here, $\alpha$ plays the role of threshold on the market weights: we only include in our generating function those stocks whose market weights exceed the threshold level $\alpha$. From \eqref{def: vartheta} and \eqref{Def: Gamma},
	\begin{equation}			\label{Eq: vartheta in alpha example}
		\vartheta_i(t) = -\mathbbm{1}_{\{\mu_i(t) \geq \alpha\}}, \qquad i=1, \cdots, d,
	\end{equation}
	and 
	\begin{equation}	\label{Eq: gamma in alpha example}
		\Gamma^G(t) = \sum_{i=1}^d L_t^{\mu_i}(\alpha).
	\end{equation}
	Note that this Gamma function is nondecreasing, and increases whenever a market weight hits the threshold $\alpha$.
	
	\smallskip
	
	The trading strategy $\varphi$, additively generated as \eqref{def: varphi}, can be represented by Proposition~\ref{prop: additive generation2} as
	\begin{equation}	\label{Eq: varphi in alpha example}
		\varphi_i(t) = -\mathbbm{1}_{\{\mu_i(t) \geq \alpha\}} + \sum_{j=1}^d \mathbbm{1}_{\{\mu_j(t) \geq \alpha\}}\mu_j(t) + V^{\varphi}(t), \qquad i = 1, \cdots, d,
	\end{equation}
	with the value
	\begin{equation*}
		V^{\varphi}(t) = 1 - \sum_{i=1}^d (\mu_i(t)-\alpha)^+ + \sum_{i=1}^d L_t^{\mu_i}(\alpha). 
	\end{equation*}
	Since the Gamma function in \eqref{Eq: gamma in alpha example} is nondecreasing, we can use the strong arbitrage condition in Theorem~\ref{thm: AGRA}: Strong relative arbitrage with respect to the market exists over every time horizon $[0, t]$ with $T_* \leq t \leq T$, satisfying
	\begin{equation*}
		\Gamma^G(T_*) = \sum_{i=1}^d L_{T_*}^{\mu_i}(\alpha) > G(\mu(0)) = 1 - \sum_{i=1}^d (\mu_i(0)-\alpha)^+.
	\end{equation*}

	\smallskip
	
	In the expression of $\varphi_i(t)$ in \eqref{Eq: varphi in alpha example}, the sum $\sum_{j=1}^d \mathbbm{1}_{\{\mu_j(t) \geq \alpha\}} + V^{\varphi}(t)$ is a universal term, same for all indices $i=1, \cdots, d$. Thus, $\varphi$ invests one currency unit less to this universal baseline amount for those `big-capitalization stocks', whose market weights exceed the threshold $\alpha$. Therefore, we can interpret the strategy $\varphi$ of \eqref{Eq: varphi in alpha example} as generating strong arbitrage relative to the market by investing more money to `small-capitalization stocks'. This is in broad agreement with previous results in Stochastic Portfolio Theory, to the effect that ``tilting'' in favor of small capitalization stocks, as opposed to their larger brethren, can lead to superior results.
\end{example}

\bigskip

\begin{example} [\textit{On the ``momentum effect"}] 	\label{Ex: momentum}
	In Example~\ref{Ex: alpha}, we compared the individual market weights $\mu_i(t)$ with a fixed constant $\alpha$, to determine whether to include them in the generating function or not. Now, we extend this idea by comparing current market weights with past market weights. To be specific, we want our trading strategy to depend on the difference between $\mu(t)$ and $\mu(t-\delta)$ for some fixed $\delta>0$.
	
	In order to do this, first we fix the time interval $\delta>0$ and enlarge the domain of each $\mu_i$ from $[0, T]$ to $[-\delta, T]$. This extension of domain can be done easily because even before we start investing to our trading strategy at time $t=0$, there must be past stock prices and past market weights. We simply attach these past data to the left of the timeline, so as to extend its domain.
	
	Furthermore, since the evolution of $\mu(t-\delta)$ is as rough as its original path $\mu(t)$, we need somehow to make it smoother. Thus, we take the moving average of market weights between very small time interval $[t-\delta,~t-\delta+\theta]$ for some small $\theta$ satisfying $0<\theta < \delta$, and use this moving average instead of $\mu(t-\delta)$. Therefore, we introduce the function of finite variation
	\begin{equation}	\label{def: A in diff example}
		A_i(t) := \frac{1}{\theta} \int_{t-\delta}^{t-\delta+\theta} \mu_i(s)ds, \qquad 0 \leq t \leq T,
	\end{equation}
	for each $i=1, \cdots, d$; this is a good estimate of $\mu_i(t-\delta)$ for some very small constant $\theta$ and $t$ fixed, and also a function of finite variation.
	
	\smallskip
	
	Now, we consider the function
	\begin{equation*}
		f(x) := \frac{1}{d}-x^+,
	\end{equation*}
	where $d$ is again the dimension of the market weight vector $\mu$. Then, for the pair $(\mu, A)$ with $A$ defined as in \eqref{def: A in diff example}, we introduce the following nonnegative generating function
	\begin{equation*}
		G(\mu(t)) = 1 - \sum_{i=1}^d \big(\mu_i(t)-A_i(t)\big)^+.
	\end{equation*}
	This generating function includes those stocks whose current market weight $\mu_i(t)$ is bigger than or equal to its (estimate of) past market weight $\mu_i(t-\delta)$. It is also very similar to that of \eqref{Def: G in alpha} with the difference that the threshold $\alpha$ is replaced by the stock-specific level $A_i(t)$, capturing in this way the ``momentum effect". 
	
	\smallskip
	
	In this manner, we compute the quantities of \eqref{def: vartheta}, \eqref{Def: Gamma} as
	\begin{equation*}
	\vartheta_i(t) = -\mathbbm{1}_{\{\mu_i(t) \geq A_i(t)\}}, \qquad i=1, \cdots, d,
	\end{equation*}
	\begin{equation}	\label{Eq: gamma in diff example}
		\Gamma^G(t) = -\sum_{i=1}^d \int_0^t \mathbbm{1}_{\{\mu_i(s) \geq A_i(s)\}}dA_i(s) + \sum_{i=1}^d L_t^{(\mu_i-A_i)}(0),
	\end{equation}
	where we recall the continuous local time $L_t^{(\mu_i-A_i)}(0)$ of $\mu_i-A_i$ at the origin, as in Definition~\ref{Def : local time}. Here, in the integral expression above, the integrand $\mathbbm{1}_{\{\mu_i(s) \geq A_i(s)\}}$ is a quantity observable at time $s$, whereas the integrator $dA_i(s)$ represents the increment of moving average of $\mu_i$ between time interval $[s-\delta,~s-\delta+\theta]$ which is also observable value at time $s$. Therefore, this integral can be computed at any time between $0$ and $T$, even though the integrand and the integrator are from different times. The last term in \eqref{Eq: gamma in diff example} is nondecreasing, but the integral term is generally not monotone, as the finite variation integrator $dA_i(s)$ in general fluctuates.
	
	\smallskip
	
	The trading strategy $\varphi$, additively generated in the manner of \eqref{def: varphi}, can be represented by Proposition~\ref{prop: additive generation2} as
	\begin{equation}	\label{Eq: varphi in diff example}
		\varphi_i(t) = -\mathbbm{1}_{\{\mu_i(t) \geq A_i(t)\}} + \sum_{j=1}^d \mathbbm{1}_{\{\mu_j(t) \geq A_j(t)\}}\mu_j(t) + V^{\varphi}(t), \qquad i = 1, \cdots, d,
	\end{equation}
	and its value is given as
	\begin{equation*}
		V^{\varphi}(t) = 1 - \sum_{i=1}^d \big(\mu_i(t)-A_i(t)\big)^+ -\sum_{i=1}^d \int_0^t \mathbbm{1}_{\{\mu_i(s) \geq A_i(s)\}}dA_i(s) + \sum_{i=1}^d L_t^{(\mu_i-A_i)}(0) L_t^{\mu_i}(\alpha). 
	\end{equation*}
	Since the Gamma function of \eqref{Eq: gamma in diff example} is no longer monotone, it is hard to formulate appropriate conditions for strong relative arbitrage in this context. We note, however, that the strategy $\varphi$ in \eqref{Eq: varphi in diff example} invests one unit of currency less in stocks whose current market weight is bigger than or equal to its (estimate of) past value.
\end{example}

\bigskip

\subsection{Construction of multiplicatively generated trading strategies}

We now recall the definitions \eqref{def: X}-\eqref{Def: Gamma} from the earlier subsection, and we further assume that $1/G\big(\mu(\cdot), A(\cdot)\big)$ is locally bounded as in Section~\ref{sec : 3.3}. Then, we consider the vector $\eta$ with components 
\begin{equation}		\label{Def: eta2}
	\eta_i(t) := \vartheta_i(t) \times \exp \Big(\int_0^{t}\frac{d\Gamma^G(s)}{G\big(\mu(s), A(s)\big)}\Big), \quad \qquad i = 1, \cdots, d, \qquad 0 \leq t \leq T
\end{equation}
in the notation \eqref{def: vartheta}, and we have the following definition as Definition~\ref{def: MG}.

\medskip

\begin{defn} [Multiplicative generation]		\label{Def : mg}
	We say that the trading strategy $\psi = (\psi_1, \cdots, \psi_d)'$ defined as
	\begin{equation}	\label{def: psi}
		\psi_i(t) := \eta_i(t) - Q^{\eta}(t) - C(0), \quad \qquad i=1, \cdots, d, \qquad 0 \leq t \leq T,
	\end{equation}
	where $Q^{\eta}$ and $C(0)$ are defined as in Definition~\ref{Def : ag}, is \textit{multiplicatively generated} by the function $G$ in \eqref{def: generating function}.
\end{defn}

\medskip

For the trading strategy $\psi$ of \eqref{def: psi}, we have the similar formula for its value as in Proposition~\ref{prop: multiplicative generation}, but the difference here is that our generating function $G$ can be a lot less smooth than before; namely, of the form \eqref{def: generating function} with absolutely continuous $f$. The proof requires additional attention and computation, as there is no `product rule' that can be applied to such less regular functions.

\medskip

\begin{prop}
	\label{prop: multiplicative generation2}
	The trading strategy $\psi$, generated multiplicatively as in \eqref{def: psi} by the function $G$ in \eqref{def: generating function} for the pair $(\mu, A)$, has value
	\begin{equation}	\label{Eq: value of MTS2}
		V^{\psi}(t):= \sum_{i=1}^d \psi_i(t) \mu_i(t) = G\big(\mu(t), A(t)\big)\exp \Big(\int_0^{t}\frac{d\Gamma^G(s)}{G\big(\mu(s), A(s)\big)}\Big) > 0, \quad 0 \leq t \leq T,
	\end{equation}
	and its components can be represented, for $i=1, \cdots, d$, in the form
	\begin{equation}
		\psi_i(t)=V^{\psi}(t)\Big(1+ \frac{1}{G\big(\mu(t), A(t)\big)} \big( \vartheta_i(t) -\sum_{j=1}^{d}\vartheta_j(t)\mu_j(t) \big) \Big). 	\label{Eq: psi2}
	\end{equation}
\end{prop}

\medskip

\begin{proof}
	We denote the exponential
	\begin{equation*}
		K(t) := \exp \Big( \int_0^t \frac{d\Gamma^G(s)}{G\big(\mu(s), A(s)\big)}\Big).
	\end{equation*}
	We recall the notation \eqref{def: X}, \eqref{def: generating function} and consider the following telescoping expansion over the refining sequence $(\mathbb{T}_n)_{n \in \mathbb{N}}$ of partitions:
	\begin{align}
		&G\big(\mu(t), A(t)\big)K(t) - G\big(\mu(0), A(0)\big)K(0)
		=\sum_{i=1}^d \sum_{\substack{[t_j, t_{j+1}] \in \mathbb{T}_n \\ t_j \leq t}} \Big\{ f_i(X_i(t_{j+1}))K(t_{j+1}) - f_i(X_i(t_{j}))K(t_{j}) \Big\}				\nonumber
		\\
		&=\sum_{i=1}^d \sum_{\substack{[t_j, t_{j+1}] \in \mathbb{T}_n \\ t_j \leq t}} \Big\{ f_i(X_i(t_{j+1}))\big(K(t_{j+1})-K(t_j)\big) \Big\}						\label{Eq: fkk}
		\\
		&~~~+\sum_{i=1}^d \sum_{\substack{[t_j, t_{j+1}] \in \mathbb{T}_n \\ t_j \leq t}} \bigg\{ \Big( f_i\big(X_i(t_{j+1})\big) - f_i\big(X_i(t_{j})\big) \Big) K(t_{j}) \bigg\}.			\label{Eq: ffk}
	\end{align}
	Then, we can further expand the last double sum \eqref{Eq: ffk} as
	\begin{align}
		&~~~~~\sum_{i=1}^d  \sum_{\substack{[t_j, t_{j+1}] \in \mathbb{T}_n \\ t_j \leq t}} \Big\{ \big( f_i(X_i(t_{j+1})) - f_i(X_i(t_{j})) \big) K(t_{j}) \Big\}			\nonumber
		\\
		&=\sum_{i=1}^d  \sum_{\substack{[t_j, t_{j+1}] \in \mathbb{T}_n \\ t_j \leq t}} \Big\{ f'_i(X_i(t_{j}))K(t_j) \big(X_i(t_{j+1}) - X_i(t_j)\big)
		+ \int_{\mathbb{R}}\mathbbm{1}_{\llparenthesis X_{t_{j}}, X_{t_{j+1}} \rrbracket}(x) |X_{t_{j+1}}-x| K(t_j)df'_i(x) \Big\}									\nonumber
		\\
		&=~~~\sum_{i=1}^d  \sum_{\substack{[t_j, t_{j+1}] \in \mathbb{T}_n \\ t_j \leq t}}  \Big\{f'_i(X_i(t_{j}))K(t_j) \big(\mu_i(t_{j+1}) - \mu_i(t_j)\big) \Big\}	\label{Eq: fkmm}
		\\
		&~~~-\sum_{i=1}^d  \sum_{\substack{[t_j, t_{j+1}] \in \mathbb{T}_n \\ t_j \leq t}}  \Big\{f'_i(X_i(t_{j}))K(t_j) \big(A_i(t_{j+1}) - A_i(t_j)\big) \Big\}		\label{Eq: fkaa}
		\\
		&~~~+\sum_{i=1}^d  \sum_{\substack{[t_j, t_{j+1}] \in \mathbb{T}_n \\ t_j \leq t}}  K(t_j) \int_{\mathbb{R}} \big(L^{X_i, \mathbb{T}_n}_{t_{j+1}}(x)-L^{X_i, \mathbb{T}_n}_{t_j}(x)\big) df'_i(x),															\label{Eq: llk}
	\end{align}
	where the first equation is from \eqref{Eq : integration by parts}, and the last follows from \eqref{def: X} and \eqref{Def: discrete local time}.
	
	\smallskip
	
	Next, we show that the sum of \eqref{Eq: fkk}, \eqref{Eq: fkaa}, and \eqref{Eq: llk} vanishes as $n \rightarrow \infty$. First, since the mesh size goes to zero as $n \rightarrow \infty$, the limit of the sum \eqref{Eq: fkk} is a Lebesgue-Stieltjes integral
	\begin{equation*}
		\sum_{i=1}^d \int_0^t f_i\big(X_i(s)\big)dK(s) = \int_0^t G\big(\mu(s), A(s)\big)dK(s),
	\end{equation*}
	because $f_i(X_i(\cdot))$ is bounded on the compact interval $[0, T]$ for each $i=1, \cdots, d$.
	From \eqref{Def: Gamma}, the change of variable formula for Lebesgue-Stieltjes integral gives
	\begin{equation}	\label{Eq: kk}
		\int_0^t G\big(\mu(s), A(s)\big)dK(s) = \int_0^t K(s)d\Gamma^G(s) = \int_0^t K(s)d\Gamma_1^G(s) - \int_0^t K(s)d\Gamma_2^G(s),
	\end{equation}
	where
	\begin{equation*}
		\Gamma^G_1(t) := \sum_{i=1}^d \int_0^t \vartheta_i(s)dA_i(s), \qquad \Gamma^G_2(t) := \sum_{i=1}^d \int_{\mathbb{R}} L_t^{X_i}(x)df'_i(x).
	\end{equation*}
	We apply the change of variable formula once again to obtain
	\begin{equation*}
		\int_0^t K(s)d\Gamma_1^G(s) = \sum_{i=1}^d \int_0^t K(s)\vartheta_i(s)dA_i(s),
	\end{equation*}
	and this is just the negative value of limit of the sum \eqref{Eq: fkaa}. On the other hand, the last integral of \eqref{Eq: kk} can be expressed as the limit of the sum
	\begin{align*}
		\int_0^t K(s)d\Gamma_2^G(s) &= \lim_{n \rightarrow \infty} \sum_{\substack{[t_j, t_{j+1}] \in \mathbb{T}_n \\ t_j \leq t}} 
		K(t_j)\big\{ \Gamma_2^G(t_{j+1}) - \Gamma_2^G(t_j) \big\}
		\\
		&=\lim_{n \rightarrow \infty} \sum_{i=1}^d  \sum_{\substack{[t_j, t_{j+1}] \in \mathbb{T}_n \\ t_j \leq t}} K(t_j)
		\int_{\mathbb{R}} \big(L_{t_{j+1}}^{X_i}(x)-L_{t_{j}}^{X_i}(x) \big) df'_i(x),
	\end{align*}
	which coincides with the limit of the sum \eqref{Eq: llk}. Therefore, the claim that the limits of the sums \eqref{Eq: fkk}, \eqref{Eq: fkaa}, and \eqref{Eq: llk} are equal to zero, is proven; whereas, the remainder term on the right-hand side of \eqref{Eq: fkk}, \eqref{Eq: ffk} is the sum \eqref{Eq: fkmm}, whose limit we denote as 
	\begin{equation*}
		\sum_{i=1}^d \int_0^t f'_i(X_i(s))K(s) d\mu_i(s) = \sum_{i=1}^d \int_0^t \eta_i(s) d\mu_i(s),
	\end{equation*}
	from \eqref{def: generating function}, and \eqref{Def: eta2}. Finally, we obtain 
	\begin{equation*}
		G\big(\mu(t), A(t)\big)K(t) - G\big(\mu(0), A(0)\big)K(0) 
		= \sum_{i=1}^d \int_0^t \eta_i(s) d\mu_i(s)
		=\sum_{i=1}^d \int_0^t \psi_i(s) d\mu_i(s),
	\end{equation*}
	where the last equation follows from the fact $\sum_{i=1}^d \mu_i(\cdot) \equiv 1$ with the construction \eqref{def: psi}. The result \eqref{Eq: value of MTS2} then follows from the self-financibility of $\psi$ and the relationship
	\begin{equation*}
		V^{\psi}(0) = \sum_{i=1}^d \psi_i(0)\mu_i(0) = \sum_{i=1}^d \big(\vartheta_i(0)-C(0)\big) \mu_i(0) = G\big(\mu(0), A(0)\big) = G\big(\mu(0), A(0)\big)K(0).
	\end{equation*}
	The equation \eqref{Eq: psi2} can be justified in the same manner as Proposition~\ref{prop: multiplicative generation}.
\end{proof}

\bigskip

\begin{example} [\textit{On the ``size effect", revisited}]	\label{Ex: alpha2}
	Recall the generating function $G$ of \eqref{Def: G in alpha} in Example~\ref{Ex: alpha}, and add a very small constant $\epsilon > 0$ to have
	\begin{equation*}
		G\big(\mu(t)\big) = (1 + \epsilon) - \sum_{i=1}^d \big(\mu_i(t)-\alpha\big)^+,
	\end{equation*}
	with the same $\vartheta$ as in \eqref{Eq: vartheta in alpha example} and the same Gamma function as in \eqref{Eq: gamma in alpha example}. The reason for inserting the constant $\epsilon > 0$ is to ensure the positivity of $G$ regardless of the choice of $\alpha \in (0, 1)$, so that $1/G$ is locally bounded.
	
	The trading strategy $\psi$, multiplicatively generated by this $G$ as in Definition~\ref{Def : mg}, can be represented by Proposition~\ref{prop: multiplicative generation2} as
	\begin{equation}	\label{Eq: psi in alpha2 example}
		\psi_i(t) = -K(t)\mathbbm{1}_{\{\mu_i(t) \geq \alpha\}} + \sum_{j=1}^d K(t) \mathbbm{1}_{\{\mu_j(t) \geq \alpha\}}\mu_j(t) + V^{\varphi}(t), \qquad i = 1, \cdots, d,
	\end{equation}
	and its value is given as
	\begin{equation*}
		V^{\psi}(t) = \Big((1+\epsilon) - \sum_{i=1}^d \big(\mu_i(t)-\alpha\big)^+ \Big) K(t),
	\end{equation*}
	where 
	\begin{equation*}
		K(t) := \exp \Big( \int_0^t \sum_{i=1}^d \frac{dL_s^{\mu_i}(\alpha)}{1+\epsilon - \sum_{j=1}^d (\mu_j(s)-\alpha)^+}  \Big).
	\end{equation*}
	From Theorem~\ref{thm: MGRA2}, strong relative arbitrage with respect to the market exists over every time horizon $[0, t]$ with $T_* \leq t \leq T$, satisfying
	\begin{equation*}
		\Gamma^G(T_*) = \sum_{i=1}^d L_{T_*}^{\mu_i}(\alpha) > (1+\epsilon) \log \Big( \frac{1 + \epsilon - \sum_{i=1}^d \big(\mu_i(0)-\alpha\big)^+}{\epsilon} \Big),
	\end{equation*}
	because $G$ satisfies the bounds $\epsilon \leq G\big(\mu(\cdot)\big) \leq 1+\epsilon$.
	
	\smallskip
	
	In the same manner as in Example~\ref{Ex: alpha}, the strategy $\psi$ in \eqref{Eq: psi in alpha2 example} invests $K(t)$ unit of currency less than the `universal baseline amount', namely $\sum_{j=1}^d K(t) \mathbbm{1}_{\{\mu_j(t) \geq \alpha\}}\mu_j(t) + V^{\varphi}(t)$, for those `big-capitalization stocks', whose market weight exceeds the threshold $\alpha$, at time $t$. Because $K(\cdot)$ is nondecreasing, the strategy $\psi$ keeps investing less and less money to those `big-capitalization stocks' as time goes by, and the ``size effect" increases gradually.
\end{example}

\bigskip

\bigskip

\bigskip

\section{Examples of entropic functions}
 \label{sec: 6}
  
In this section, we present some examples of trading strategies additively and multiplicatively generated from variants of the `entropy function', and the corresponding conditions for strong relative arbitrage introduced in Section~\ref{sec: 4}. Empirical results regarding these examples will be presented in the next section.

\bigskip

Consider the Gibbs entropy function 
\begin{equation}	\label{Def: Gibbs}
	H(x) = -\sum_{i=1}^d x_i \log (x_i), \qquad x \in (0, 1)^d,
\end{equation}
with values in (0, $\log d$). Being nonnegative, twice-differentiable and concave, this function is one of the most frequently used functions in stochastic portfolio theory. See \citet{Fe, FK_survey, Karatzas:Ruf:2017} for its usage in generating portfolios, and also \cite{Ruf:Xie}, \cite{Schied:2016} for some variants of portfolios generated by this function. 

\medskip

\begin{example} [Entropy function]
	\label{Ex: original entropy}
	In order to compare the trading strategy generated by the original entropy function, with those generated from variants of functions related to it, we first derive and summarize the trading strategy additively/multiplicatively generated by the original entropy function. Consider the ``shifted entropy"
	\begin{equation}	\label{Eq: original entropy}
		G\big(\mu(t)\big):=-\sum_{i=1}^d\mu_i(t)\log \big(p\mu_i(t) \big) = - \log p -\sum_{i=1}^d\mu_i(t)\log\big(\mu_i(t)\big),
	\end{equation}
	for some given real constant $p \geq 1$, where the last equality uses the fact $\sum_{i=1}^d \mu_i(t) = 1$. This quantity coincides with the original entropy $H\big(\mu(t)\big)$ in \eqref{Def: Gibbs} when $p=1$; the reason for inserting the additive constant will be explained in the following remark. From \eqref{Eq: gamma}, \eqref{Eq: varphi}, and \eqref{Eq: psi}, the additively generated trading strategy $\varphi$, and the multiplicatively generated trading strategy $\psi$ from this entropy function, can be represented as
	\begin{equation}	\label{Eq: varphi of original entropy}
		\varphi_i(t)= - \log \big(p\mu_i(t)\big)+\Gamma^G(t), \qquad i=1, \cdots, d,
	\end{equation}
	\begin{equation}	\label{Eq: psi of original entropy}
		\psi_i(t)= - \exp \bigg( \int_0^t \frac{d\Gamma^G(s)}{G\big(\mu(s)\big)} \bigg) \log \big(p\mu_i(t)\big), \qquad i=1, \cdots, d,
	\end{equation}
	where
	\begin{equation}		\label{Eq: gamma of original entropy}
		\Gamma^G(t)=\sum_{i=1}^d \int_0^t \frac{d\langle \mu_i \rangle (s)}{2\mu_i(s)}
	\end{equation}
	is nondecreasing in $t$. The values of these trading strategies are given via \eqref{Eq: value of ATS} and \eqref{Eq: value of MTS}. Note that $\varphi$ in \eqref{Eq: varphi of original entropy} and $\psi$ in \eqref{Eq: psi of original entropy} have relatively simple forms, because $G$ in \eqref{Eq: original entropy} is `almost balanced', in the sense that
	\begin{equation*}
		G(\mu(\cdot)) - 1 = \sum_{j=1}^{d}\mu_j(\cdot)\partial_jG(\mu(\cdot))
	\end{equation*}
	holds; compare this equation with \eqref{Def: balance}, and also compare \eqref{Eq: varphi of original entropy}, \eqref{Eq: psi of original entropy} with \eqref{Eq: balanced varphi} and \eqref{Eq: balanced psi}. Then, the condition \eqref{Con: additive arb} for additively generated strong arbitrage in Theorem~\ref{thm: AGRA} is given as
	\begin{equation}	\label{Con: additive arb of original entropy}
		\sum_{i=1}^d \int_0^{T_*} \frac{d\langle \mu_i \rangle (s)}{2\mu_i(s)} > -\sum_{i=1}^d\mu_i(0)\log \big( p\mu_i(0) \big),
	\end{equation}
	whereas the condition \eqref{Con: multiplicative arb2} for multiplicatively generated strong arbitrage in Theorem~\ref{thm: MGRA2} is
	\begin{equation}	\label{Con: multiplicative arb of original entropy}
		\sum_{i=1}^d \int_0^{T_*} \frac{d\langle \mu_i \rangle (s)}{2\mu_i(s)} > \beta \log \bigg( \frac{-\sum_{i=1}^d\mu_i(0)\log \big( p\mu_i(0) \big)}{\alpha} \bigg).
	\end{equation}
	Here, the constants $\alpha$, $\beta$ are the lower and upper bounds on $G$, which appear in the boundedness condition \textit{(i)} of Theorem~\ref{thm: MGRA2}. We discuss these bounds on $G$ in the Remark~\ref{Re: Reducing threshold in entropy example} below.
\end{example}

\medskip

\begin{rem}
	\label{Re: Reducing threshold in entropy example}
	The construction of trading strategies described in the previous sections does not require any optimization or statistical estimation of parameters. However, we can improve the relative performance of trading strategies with respect to the market by introducing a parameter or a set of parameters in the generating function $G$. Though the original entropy function is as in \eqref{Eq: original entropy} with $p=1$, we purposely inserted a constant $p$ inside the logarithm. To achieve strong relative arbitrage faster, or to find the smallest such $T_*$ satisfying \eqref{Con: additive arb of original entropy}, or more generally \eqref{Con: additive arb}, it helps to be able to make the `threshold' value $G\big(\mu(0), A(0)\big)$ on the right-hand side of the inequality smaller, while keeping the `growth rate' of $\Gamma^G(\cdot)$ fixed.

	\smallskip	
	
	It is in this spirit, that we placed the parameter $p$ in \eqref{Eq: original entropy}; inserting such constant $p>1$ inside the $\log$ would make the initial value $G\big(\mu(0)\big)$ smaller by the amount $\log p$, compared to the case $p=1$, while it does not affect $\Gamma^G(\cdot)$, as subtracting a constant $\log p$ from $G$ does not change any derivatives of $G$ with respect to the market weights. However, if $p$ is so large that $-\sum_{i=1}^d\mu_i(t)\log\mu_i(t) < \log p$ holds at some time $t$, then $G\big(\mu(t), A(t)\big)$ has a negative value. Theoretically, $-\sum_{i=1}^d\mu_i(t)\log\mu_i(t)$ has the minimum value of $0$ only when one of the market weights, say $\mu_1(t)$, is equal to $1$, and all the other weights $\mu_i(t)$ for $i=2, \cdots, d$ vanish, which does not happen in the real world. Empirically, the value of $-\sum_{i=1}^d\mu_i(t)\log\mu_i(t)$ is always bounded away from zero, and we can guarantee this condition theoretically by imposing a weak condition on the market weights. For example, restricting the maximum value of the market weights, say 
	\begin{equation}	\label{Con: max market weight}
		\max_{i} \mu_i(\cdot) \leq 0.5
	\end{equation}
	yields an additional condition on the market weights, namely; there must be an index $j \in \{1, \cdots, d\}$ such that 
	\begin{equation}	\label{Con: min market weight}
		\mu_j(t) \geq \frac{0.5}{d-1},
	\end{equation}
	for any $t \in [0, T]$, thanks to the identity $\sum_{i=1}^d \mu_i \equiv 1$. Then, the value of $-\sum_{i=1}^d\mu_i(t)\log\mu_i(t)$ should be bigger than $-\frac{0.5}{d-1} \log \big(\frac{0.5}{d-1}\big)$, and is bounded away from $0$ at all times. Finding a suitable value of $p>1$, while maintaining $G$ bounded away from $0$ (and bigger than some positive constant $\alpha$) should be statistically done and it depends on $d$, the number of stocks. It is quite straightforward that $G$ is bounded from above by some constant $\beta$, as the function $x \mapsto -x\log x$ has the maximum value $1/e$. Empirical estimation of such $p$ can be found in the next section.
	
	\smallskip
	
	Making the initial value of $G\big(\mu(0), A(0)\big)$ small while keeping the growth rate of $\Gamma^G(\cdot)$ is also beneficial for calculating the `excess return rate' of trading strategies with respect to the market. The excess return rate of the trading strategy $\varphi$ at time $t \in (0, T]$ can be defined as
	\begin{equation}		\label{Def: excess return rate}
		R^\varphi(t) := \frac{V^{\varphi}(t) - V^{\varphi}(0)}{V^{\varphi}(0)},
	\end{equation}
	and from \eqref{Eq: value of ATS}, this can be represented as
	\begin{equation*}
			R^\varphi(t) = \frac{G\big(\mu(t), A(t)\big) + \Gamma^G(t) - G\big(\mu(0), A(0)\big)}{G\big(\mu(0), A(0)\big)},
	\end{equation*}
	in the case of additively generated trading strategy. Thus, if we somehow make the value $G\big(\mu(0), A(0)\big)$, the denominator of above fraction, smaller, while keeping the value of $\Gamma^G(t)$ in the numerator, we can obtain larger excess return rates for the trading strategy $\varphi$. In the following examples, we use this trick to decrease the initial value $G\big(\mu(0), A(0)\big)$ of generating function by inserting an appropriate constant $p$ whenever possible.
\end{rem}

\medskip

The following two examples use for the component $A$ two ``polar opposite" functions of finite variation; the running maximum 
\begin{equation}		\label{Def : max}
	\mu_i^*(t) := \max_{0 \leq s \leq t} \mu_i(s),	
\end{equation}
and the running minimum 
\begin{equation}		\label{Def : min}
	\mu_{*i}(t) := \min_{0 \leq s \leq t} \mu_i(s),
\end{equation}
of the market weights. 

\medskip

\begin{example} [Entropy function with running maximum]
	\label{Ex: entropy with MAX}
	Consider an entropic function of the type
	\begin{equation}		\label{Def: entropy with MAX}
		G\big(\mu(t), A(t)\big) \equiv G\big(\mu(t), \mu^*(t)\big) := - \log p - \sum_{i=1}^d \mu_i(t)\log \mu_i^*(t),
	\end{equation}
	with the notation of the vector function $A \equiv \mu^* = (\mu_1^*, \cdots, \mu_d^*)'$. As before, $p \geq 1$ is a constant as in Remark~\ref{Re: Reducing threshold in entropy example}, and the initial value $G\big(\mu(0), \mu^*(0)\big) = - \log p - \sum_{i=1}^d \mu_i(0)\log \mu_i(0) $ is the same as in Example~\ref{Ex: original entropy}. We then easily obtain the derivatives as
	\begin{equation*}
		\partial_iG\big(\mu(t), \mu^*(t)\big) = -\log \mu_i^*(t), \qquad
		\partial^2_{i, j}G\big(\mu(t), \mu^*(t)\big) = 0, \qquad
		D_iG\big(\mu(t), \mu^*(t)\big) = - \frac{\mu_i(t)}{\mu_i^*(t)},
	\end{equation*}
	for $1 \leq i,~ j \leq d$.
	From \eqref{Eq: gamma}, we also have
	\begin{equation}					\label{Eq: gamma of entropy with MAX}
		\Gamma^G(t) = \sum_{i=1}^{d}\int_0^t \frac{\mu_i(s)}{\mu_i^*(s)} d \mu_i^*(s) 
		= \sum_{i=1}^d \big( \mu_i^*(t) - \mu_i(0) \big) = \sum_{i=1}^d \mu_i^*(t) - 1,	
	\end{equation}
	where we used the fact that the increment $d\mu_i^*(s)$ is positive only when $\mu_i(s) = \mu_i^*(s)$. As the function $G$ of \eqref{Def: entropy with MAX} is linear in $\mu_i(\cdot)$, the second order partial derivatives with respect to $\mu_i$ of $G$ vanish, and the nondecreasing structure of $\Gamma^G(\cdot)$ comes solely from $\mu_i^*(\cdot)$. Also from \eqref{Eq: value of ATS}, and \eqref{Eq: varphi}, the trading strategy $\varphi$ generated additively from this function in \eqref{Def: entropy with MAX}, is expressed as
	\begin{equation}	\label{Eq: varphi of entropy with MAX}
		\varphi_i(t)=-\log \big(p\mu_i^*(t)\big) + \sum_{j=1}^d \mu_j^*(t) - 1, \qquad i=1, \cdots, d;
	\end{equation}
	and the value of this trading strategy is given as
	\begin{equation*}
		V^{\varphi}(t) =-\sum_{i=1}^d \mu_i(t)\log \big(p\mu_i^*(t)\big)+\sum_{i=1}^d \mu_i^*(t)-1.
	\end{equation*}
	The strong arbitrage condition \eqref{Con: additive arb} in Theorem~\ref{thm: AGRA} takes the form
	\begin{equation*}
		\sum_{i=1}^d \mu_i^*(T_*) > 1- \sum_{i=1}^d\mu_i(0)\log\big(p\mu_i(0)\big).
	\end{equation*}
	
	\smallskip
	
	On the other hand, from \eqref{Eq: value of MTS}, and \eqref{Eq: psi}, the trading strategy $\psi$ generated multiplicatively by the function in \eqref{Def: entropy with MAX}, is given as
	\begin{equation}	\label{Eq: psi of entropy with MAX}
		\psi_i(t)= -K(t) \log \big(p\mu_i^*(t)\big), \qquad i=1, \cdots, d;
	\end{equation}
	and the associated value is
	\begin{equation*}
		V^{\psi}(t) =-K(t) \sum_{i=1}^d \mu_i(t)\log \big(p\mu_i^*(t)\big),
	\end{equation*}
	where
	\begin{equation*}
		K(t) :=\exp \bigg( - \int_0^t \sum_{i=1}^d \frac{d\mu_i^*(s)}{\sum_{j=1}^d \mu_j(s)\log \big(p\mu_j^*(s)\big)} \bigg).
	\end{equation*}
	The strong arbitrage condition \eqref{Con: multiplicative arb2} in Theorem~\ref{thm: MGRA2} takes the form
	\begin{equation*}
		\sum_{i=1}^d \mu_i^*(T_*) > 1 + \beta \log \bigg( \frac{-\sum_{i=1}^d\mu_i(0)\log\big(p\mu_i(0)\big)}{\alpha} \bigg).
	\end{equation*}
	Here $\alpha, \beta$ are again lower and upper bounds on $G$, and these bounds depend on the parameter $p$ and the condition imposed on the market weights, as described in Remark~\ref{Re: Reducing threshold in entropy example}.	Empirical results regarding this example can be found in the next section.	
\end{example}

\bigskip

The Gamma function $\Gamma^G(\cdot)$ which represents the ``cumulative earnings" of the next example is nonincreasing, but surprisingly, the empirical value $V^{\varphi}(\cdot)$ and $V^{\psi}(\cdot)$ of trading strategies grow asymptotically in the long run as the value of $G$ grows, as indicated in the empirical results of the next section. Thus, in this case, it is more appropriate to apply Theorem~\ref{thm: AGRA2} and Theorem~\ref{thm: MGRA3} regarding the strong arbitrage condition.

\bigskip

\begin{example} [Entropy function with running minimum]
	\label{Ex: entropy with MIN}
	Consider the function
	\begin{equation}		\label{Def: entropy with MIN}
		G\big(\mu(t), A(t)\big) \equiv G\big(\mu(t), \mu_*(t)\big) := - \log p - \sum_{i=1}^d \mu_i(t)\log \mu_{*i}(t),
	\end{equation}
	with the notation of the vector function $A \equiv \mu_* = (\mu_{*1}, \cdots, \mu_{*d})'$ in \eqref{Def : min}. As before, $p$ is a constant and the initial value $G\big(\mu(0), \mu_*(0)\big)$ is the same as previous examples. Then, similarly as before, we have
	\begin{equation*}
		\partial_iG\big(\mu(t), \mu_*(t)\big) = -\log \mu_{*i}(t), \qquad
		\partial^2_{i, j}G\big(\mu(t), \mu_*(t)\big) = 0, \qquad
		D_iG\big(\mu(t), \mu_*(t)\big) = - \frac{\mu_i(t)}{\mu_{*i}(t)},
	\end{equation*}
	for $1 \leq i, j \leq d$.
	Also from \eqref{Eq: gamma}, we obtain
	\begin{equation}		\label{Eq: gamma of entropy with MIN}
		\Gamma^G(t) = \sum_{i=1}^{d}\int_0^t \frac{\mu_i(s)}{\mu_{*i}(s)} d \mu_{*i}(s)
		=\sum_{i=1}^{d}\int_0^t ~1~ d \mu_{*i}(s)
		=\sum_{i=1}^d \mu_{*i}(t) - 1,		
	\end{equation}
	which is nonpositive and nonincreasing function of $t$.
	
	\smallskip
	
	We first consider the trading strategy $\varphi$ additively generated from this function which is expressed as
	\begin{equation}	\label{Eq: varphi of entropy with MIN}
		\varphi_i(t)=-\log \big(p\mu_{*i}(t)\big) + \sum_{j=1}^d \mu_{*j}(t)-1, \qquad i=1, \cdots, d,
	\end{equation}
	by \eqref{Eq: varphi}. Note that $\varphi_i(t)$ admits the lower bound
	\begin{equation}			\label{Eq: varphi estimate in entropy with MIN}
		\varphi_i(t) =-\log p -\log \mu_{*i}(t) + \mu_{*i}(t) + \sum_{\substack{j=1 \\ j \ne i}}^d \mu_{*j}(t) - 1				
		\geq -\log p -\log \mu_i(0) + \mu_i(0) -1,
	\end{equation}
	because the function $x \mapsto -\log x + x$ is decreasing in the interval $x \in (0, 1)$ and, thus, the quantity $\varphi_i(t)$ is positive provided that  
	\begin{equation*}
		p < e^{-\log\mu_i(0)+\mu_i(0)-1}
	\end{equation*}
	holds. By \eqref{Eq: value of ATS}, the value of this trading strategy is given as
	\begin{equation}		\label{Eq: value of entropy with MIN}
		V^{\varphi}(t) =-\log p - \sum_{i=1}^d \mu_i(t)\log \mu_{*i}(t) + \Big( \sum_{i=1}^d \mu_{*i}(t) -1 \Big).
	\end{equation}
	While $\Gamma^G(t) = \sum_{i=1}^d \mu_{*i}(t)-1$, the last term on the right-hand side of \eqref{Eq: value of entropy with MIN}, is nonincreasing, the second term $-\sum_{i=1}^d \mu_i(t)\log \mu_{*i}(t)$ asymptotically increases as the mapping $t \mapsto -\log \mu_{*i}(t)$ is nondecreasing. Actually, as we can see in the next section, the value of this trading strategy grows in the long run. We can apply Theorem~\ref{thm: AGRA2}, rather than Theorem~\ref{thm: AGRA}, to find a strong arbitrage condition, because $\Gamma^G(\cdot)$ in this example is not nondecreasing.
	
	\smallskip
	
	In order to apply Theorem~\ref{thm: AGRA2}, we first need to show that $V^{\varphi}(\cdot) \geq 0$ holds. From \eqref{Eq: varphi estimate in entropy with MIN}, we obtain
	\begin{align*}
		-\log \mu_{*i}(t) &\geq -\sum_{j=1}^d\mu_{*i}(t)-\log \mu_i(0)+\mu_i(0)
		\geq -1 -\log \mu_i(0)+\mu_i(0)
		\\
		&\geq -1 -\log \big( \max_{j=1, \cdots, d}\mu_j(0) \big) +\max_{j=1, \cdots, d}\mu_j(0)
	\end{align*}
	holds for all $i=1, \cdots, d$. The last inequality follows from the fact that the function $x \mapsto -\log x + x$ is decreasing in the interval $x \in [0, 1]$. Then, we also obtain
	\begin{equation*}
		-\sum_{i=1}^d \mu_i(t)\log \mu_{*i}(t) \geq -1 -\log \big( \max_{j=1, \cdots, d}\mu_j(0) \big) +\max_{j=1, \cdots, d}\mu_j(0),
	\end{equation*}
	because $- \sum_{i=1}^d \mu_i(t)\log \mu_{*i}(t)$ is just the weighted arithmetic average of $\{-\log \mu_{*i}(t)\}_{i=1, \cdots, d}$ with weights $\mu_i(t)$ with $\sum_{i=1}^d \mu_i(t)=1$. Thus, $V^{\varphi}(t)$ in \eqref{Eq: value of entropy with MIN} admits the lower bound
	\begin{align*}
		V^{\varphi}(t) \geq -\log p -2 -\log \big( \max_{j}\mu_j(0) \big) +\max_{j}\mu_j(0)
	\end{align*}
	for any $t \in [0, T]$, and $V^{\varphi}(\cdot) \geq 0$ is guaranteed when
	\begin{equation}		\label{Con: p of entropy with MIN}
		p \leq e^{-2-\log \big( \max_{j}\mu_j(0) \big) +\max_{j}\mu_j(0)}
	\end{equation}
	holds. Regarding the second condition of Theorem~\ref{thm: AGRA2}, we have
	\begin{align}
		G\big(\mu(t), \mu_*(t)\big) &= - \log p - \sum_{i=1}^d \mu_i(t)\log \mu_{*i}(t)	\nonumber
		\\
		& \geq - \log p - \sum_{i=1}^d \mu_i(t)\log \big( \max_{i=1, \cdots, d} (\mu_{*i}(t) ) \big)	\nonumber
		\\
		& = -\log p - \max_{i=1, \cdots, d} \big\{ \log \mu_{*i}(t) \big\}
		:= F\big(\mu(t), \mu_*(t)\big),	\label{Eq: F of entropy with MIN}
	\end{align}
	where we used the fact $\sum_{i=1}^d \mu_i(t) = 1$; now the mapping $t \mapsto \mu_{*i}(t)$ is nonincreasing, so $F\big(\mu(t), \mu_*(t)\big)$ is nondecreasing in $t$. Finally, the last condition of Theorem~\ref{thm: AGRA2} follows easily from \eqref{Eq: gamma of entropy with MIN}, as 
	\begin{equation}	\label{Eq: kappa of entropy with MIN}
		\Gamma^G(t) \geq -1 := -\kappa.
	\end{equation}
	Thus, Theorem~\ref{thm: AGRA2} shows that the additively generated strategy $\varphi$ in \eqref{Eq: varphi of entropy with MIN} is strong arbitrage relative to the market over every time horizon $[0, t]$ with $T_* \leq t \leq T$, satisfying the condition
	\begin{equation*}
		\sum_{i=1}^d \mu_i(0)\log \mu_{i}(0) - \max_{i=1, \cdots, d} \big\{ \log \mu_{*i}(T_*) \big\} > 1.
	\end{equation*}
	
	\smallskip
	
	Next, from \eqref{Eq: psi}, the trading strategy $\psi$ multiplicatively generated by the function \eqref{Def: entropy with MIN} is represented as
	\begin{equation}	\label{Eq: psi of entropy with MIN}
		\psi_i(t)= -K(t) \log \big(p\mu_{*i}(t)\big), \qquad i=1, \cdots, d;
	\end{equation}
	with the value
	\begin{equation*}
	V^{\psi}(t) =-K(t) \sum_{i=1}^d \mu_i(t)\log \big(p\mu_{*i}(t)\big),
	\end{equation*}
	where
	\begin{equation*}
	K(t) :=\exp \bigg( - \int_0^t \sum_{i=1}^d \frac{d\mu_{*i}(s)}{\sum_{j=1}^d \mu_j(s)\log \big(p\mu_{*j}(s)\big)} \bigg).
	\end{equation*}
	
	\smallskip
	
	For the strong arbitrage condition, we use Theorem~\ref{thm: MGRA3}. Since $F\big(\mu(t), \mu_*(t)\big)$ and $\kappa$ defined in \eqref{Eq: F of entropy with MIN}, \eqref{Eq: kappa of entropy with MIN} satisfy the conditions \textit{(i), (ii)} (with an appropriate choice of $p$ to make $F$ a positive function), the strong arbitrage condition \eqref{Con: multiplicative arb3} becomes
	\begin{equation*}
	\log \Big(-\max_{i=1, \cdots, d} \big\{ \log p\mu_{*i}(T_*) \big\}\Big) > \frac{-1}{\log p + \max_{i=1, \cdots, d} \big\{ \log \mu_{i}(0) \big\}}.
	\end{equation*}
\end{example}

\bigskip

\begin{rem}
	\label{Re: Reducing threshold in entropy example2}
	In Remark~\ref{Re: Reducing threshold in entropy example}, we need to find a suitable value for $p$ satisfying an inequality, for instance, $-\sum_{i=1}^d\mu_i(t)\log(\mu_i(t)) \geq \log p$~ for all $t \in [0, T]$ in Example~\ref{Ex: original entropy}, to make the function $G$ nonnegative. This inequality usually depends on the values $\mu_i(t)$, $t \in [0, T]$ which are not observable at time $0$. Thus, we need to impose some condition on the market weights or statistically analyze historical market data to find an appropriate value for $p$ before we construct the trading strategy.
	
	\smallskip
	
	However, in Example~\ref{Ex: entropy with MIN}, due to its unique structure, we can analytically find a suitable value of $p$ without any statistical estimation at time $t=0$. Indeed, from \eqref{Eq: F of entropy with MIN}, we have that
	\begin{align*}
		G\big(\mu(t), \mu_*(t)\big) & \geq -\log p - \max_{i=1, \cdots, d} \big\{ \log \mu_{*i}(t) \big\}
		\\
		& \geq -\log p - \max_{i=1, \cdots, d} \big\{ \log \mu_{i}(0) \big\}
	\end{align*}
	holds; and setting 
	\begin{equation}		\label{Con: p of entropy with MIN2}
		p = \frac{1}{\max_{i=1, \cdots, d}\mu_{i}(0)}
	\end{equation}
	guarantees the condition $G\big(\mu(t), \mu_*(t)\big) \geq 0$ for all $t \in [0, T]$. Note that this $p$ can be calculated from absolutely observable values at time $0$. Actually, $p$ satisfying \eqref{Con: p of entropy with MIN} also guarantees the nonnegativity condition of $G$ because $G\big(\mu(\cdot), \mu_*(\cdot)\big) \geq V^{\varphi}(\cdot) = G\big(\mu(\cdot), \mu_*(\cdot)\big)+\Gamma^G(\cdot) \geq 0$ holds due to the nonpositivity of $\Gamma^G(\cdot)$. Of course, one can perform a statistical estimation of $p$ using past market data, to obtain a better value of $p$ while satisfying both $G\big(\mu(\cdot), \mu_*(\cdot)\big) \geq 0$ and $V^{\varphi}(\cdot) \geq 0$.
\end{rem}

\bigskip

The next example provides yet another application of Theorem~\ref{thm: AGRA2}.

\bigskip

\begin{example} [Iterated entropy function with running minimum]
	\label{Ex: entropy with LOGLOG}
	In this example, we first fix a positive constant $r$ such that the following condition on the initial market weights holds;
	\begin{equation} \label{Con: initial condition in entropy with LOGLOG}
		\mu_i(0) \leq \frac{1}{re}, \qquad i = 1, \cdots, d.
	\end{equation}
	Here, $e$ is the exponential constant. As the initial market weights are observable before we construct a trading strategy, we can find and fix such value of $r$ at the moment we start investing in our trading strategy. For example, if no single stock takes more than $12\%$ of total capitalization at time $0$, we can set $r=3$, as $\frac{1}{3e} \approx 0.123$. Then, we consider a function
	\begin{equation}		\label{Def: entropy with LOGLOG}
		G\big(\mu(t), A(t)\big) \equiv G\big(\mu(t), \mu_*(t)\big) := - p - \sum_{i=1}^d \mu_i(t)\log \big\{ -r\mu_{*i}(t)\log \big(r\mu_{*i}(t)\big) \big\},
	\end{equation}
	with the notation of the vector function $A \equiv \mu_* = (\mu_{*1}, \cdots, \mu_{*d})'$. As in Remark~\ref{Re: Reducing threshold in entropy example2}, we can pre-determine the value of the constant $p$, without any statistical estimation, because of the series of inequalities
	\begin{align}
		G\big(\mu(t), \mu_*(t)\big) & \geq -p - \sum_{i=1}^d \mu_i(t) \log \Big[ \max_{j=1, \cdots, d} \big\{ -r\mu_{*j}(t)\log \big(r\mu_{*j}(t)\big) \big\} \Big]		\nonumber
		\\
		& \geq -p - \log \Big[ \max_{j=1, \cdots, d} \big\{ -r\mu_{*j}(t)\log \big( r\mu_{*j}(t) \big) \big\} \Big]	 =: F\big(\mu(t), \mu_*(t)\big)					\label{Eq: F of entropy with LOGLOG}
		\\
		& \geq -p - \log \Big[ \max_{j=1, \cdots, d} \big\{ -r\mu_{j}(0)\log \big(r\mu_{j}(0) \big) \big\} \Big], \qquad \forall ~ t \in [0, T].				\nonumber
	\end{align}
	The first inequality uses the fact that $x \mapsto -\log x$ is decreasing function and the second inequality is from the equation $\sum_{i=1}^d \mu_i(t) = 1$. The last inequality holds because $x \mapsto -rx \log (rx)$ is increasing in the interval $[0, \frac{1}{re}]$ and
	\begin{equation}	\label{Con: initial condition in entropy with LOGLOG2}
		0 \leq \mu_{*i}(\cdot) \leq \mu_i(0) \leq \frac{1}{re},
	\end{equation}
	holds from the assumption \eqref{Con: initial condition in entropy with LOGLOG}. Note that $F\big(\mu(t), \mu_*(t)\big)$ defined in \eqref{Eq: F of entropy with LOGLOG} is a nondecreasing in $t$ as the mappings $t \mapsto \mu_{*i}(t)$ and $t \mapsto -r\mu_{*i}(t)\log \big(r\mu_{*i}(t)\big)$ are nonincreasing. Then, the choice
	\begin{equation}	\label{Eq: p in entropy with LOGLOG}
		p \leq - \log \Big[ \max_{i=1, \cdots, d} \big\{ -r\mu_{i}(0)\log \big( r\mu_{i}(0) \big) \big\} \Big],
	\end{equation}
	which is completely observable value at time $0$, guarantees that $G\big(\mu(\cdot), \mu_*(\cdot)\big)$ is always nonnegative. Next, after some computation, we obtain the partial derivatives
	\begin{align}
		&\partial_iG\big(\mu(t), \mu_*(t)\big) = -\log \big\{ -r\mu_{*i}(t)\log \big(r\mu_{*i}(t)\big) \big\} \geq 1,	\label{Con: partial in entropy with LOGLOG}
		\\
		&\partial^2_{i, k}G\big(\mu(t), \mu_*(t)\big) = 0,		\nonumber
		\\
		&D_iG\big(\mu(t), \mu_*(t)\big) = -\frac{\mu_i(t) \log \big( r\mu_{*i}(t) \big)+\mu_i(t)}{\mu_{*i}(t)\log \big(r\mu_{*i}(t)\big)},		\nonumber
	\end{align}
	for $1 \leq i, k \leq d$. We note that $\partial_iG\big(\mu(t), \mu_*(t)\big) \geq 1 $ holds again because the mapping $x \mapsto -rx \log (rx)$ is increasing from $0$ to $\frac{1}{e}$ in the interval $[0, \frac{1}{re}]$. From \eqref{Eq: gamma} and the fact that the increment $d\mu_{*i}(s)$ is positive only when $\mu_i(s) = \mu_{*i}(s)$, we obtain
	\begin{equation}		\label{Eq: gamma of entropy with LOGLOG}
		\Gamma^G(t) = \sum_{i=1}^{d}\int_0^t \Big( 1+\frac{1}{\log \big(r\mu_{*i}(s)\big)} \Big) d \mu_{*i}(s)
	\end{equation}
	which is nonincreasing function of $t$, because $0 \leq 1+\frac{1}{\log \big(r\mu_{*i}(\cdot)\big)} \leq 1$ holds by the equation \eqref{Con: initial condition in entropy with LOGLOG2}. This function admits the lower bound
	\begin{align}
		\Gamma^G(t) &= \sum_{i=1}^{d}\int_0^t ~1 ~d \mu_{*i}(s) + \sum_{i=1}^{d}\int_0^t \frac{1}{\log \big(r\mu_{*i}(s)\big)} d \mu_{*i}(s)		\nonumber
		\\
		&= \sum_{i=1}^{d} \mu_{*i}(t) - \sum_{i=1}^{d} \mu_{*i}(0) + \sum_{i=1}^{d} li_r(\mu_{*i}(t)) - \sum_{i=1}^{d} li_r(\mu_{i}(0)),		\nonumber
		\\
		& \geq -1 -\sum_{i=1}^{d} li_r(\mu_{i}(0)) =: -\kappa,		\label{Eq: kappa of entropy with LOGLOG}
	\end{align}
	with the notation
	\begin{equation*}
		li_r(x):=\int_0^x \frac{du}{\log (ru)}=\frac{1}{r}\int_0^{rx}\frac{dv}{\log v} = \frac{1}{r}li(rx).
	\end{equation*}
	Here, $li(x) = \int_0^x \frac{du}{\log u}$ represents the logarithmic integral function. Note that the function $li_r(x)$ has negative value and is decreasing from $0$ to $-\infty$ in the interval $x \in [0, \frac{1}{r})$. The last inequality holds because the inequality 
	\begin{equation}		\label{Eq: li estimate}
		\mu_{*i}(\cdot) + li_r(\mu_{*i}(\cdot)) \geq 0	
	\end{equation}
	is satisfied for all $\mu_{*i}(\cdot)$ with the condition \eqref{Con: initial condition in entropy with LOGLOG2}. We also note that $\kappa$ defined in \eqref{Eq: kappa of entropy with LOGLOG} satisfies $-1+\sum_{i=1}^d\mu_i(0) = 0 \leq \kappa < 1$ from the same inequality \eqref{Eq: li estimate}. On the other hand, by \eqref{Eq: varphi}, the trading strategy $\varphi$ additively generated from this function is expressed as
	\begin{equation}	\label{Eq: varphi of entropy with LOGLOG}
		\varphi_i(t)= -p -\log \big\{ -r\mu_{*i}(t)\log \big(r\mu_{*i}(t)\big) \big\} + \sum_{i=1}^{d}\int_0^t \Big( 1+\frac{1}{\log \big(r\mu_{*i}(s)\big)} \Big) d \mu_{*i}(s).
	\end{equation}
	Finally, by \eqref{Eq: value of ATS}, the value of this trading strategy $\varphi$ is given as
	\begin{equation}		\label{Eq: value of entropy with LOGLOG}
		V^{\varphi}(t) =- p - \sum_{i=1}^d \mu_i(t)\log \big\{ -r\mu_{*i}(t)\log \big(r\mu_{*i}(t)\big) \big\} + \sum_{i=1}^{d}\int_0^t \Big( 1+\frac{1}{\log \big(r\mu_{*i}(s)\big)} \Big) d \mu_{*i}(s),
	\end{equation}
	and is estimated as
	\begin{equation*}
		V^{\varphi}(t) \geq -p- \log \Big[ \max_{i=1, \cdots, d} \big\{ -r\mu_{i}(0)\log \big(r\mu_{i}(0)\big) \big\} \Big] - \kappa,
	\end{equation*}
	from \eqref{Eq: F of entropy with LOGLOG} and \eqref{Eq: kappa of entropy with LOGLOG}. Thus, the choice 
	\begin{equation}			\label{Eq: p in entropy with LOGLOG2}
		p = - \log \Big[ \max_{i=1, \cdots, d} \big\{ -r\mu_{i}(0)\log \big(r\mu_{i}(0)\big) \big\} \Big] - \kappa
	\end{equation}
	guarantees $V^{\varphi}(\cdot) \geq 0$ and also satisfies \eqref{Eq: p in entropy with LOGLOG}. We emphasize here again that $p$ defined as in \eqref{Eq: p in entropy with LOGLOG2} depends only on the initial market weights $\mu_i(0)$, thus no statistical estimation of $p$ is required. Using the same technique as in \eqref{Eq: F of entropy with LOGLOG}, $\varphi_i(t)$ in \eqref{Eq: varphi of entropy with LOGLOG} is greater or equal to 
	\begin{equation*}
		-p - \log \Big[ \max_{i} \big\{ -r\mu_{i}(0)\log \big(r\mu_{i}(0) \big) \big\} \Big] - \kappa,
	\end{equation*}
	which is $0$ by \eqref{Eq: p in entropy with LOGLOG2}. Thus, this trading strategy is `long-only', i.e., $\varphi_i(\cdot) \geq 0$ for all $i=1, \cdots, d$.
	
	\smallskip
	
	As we showed above that all conditions of Theorem~\ref{thm: AGRA2} are satisfied, the additively generated strategy $\varphi$ in \eqref{Eq: varphi of entropy with LOGLOG} is strong arbitrage relative to the market over every time horizon $[0, t]$ with $T_* \leq t \leq T$, satisfying the condition
	\begin{equation*}
		- \log \Big[ \max_{i=1, \cdots, d} \big\{ -r\mu_{i}(T_*)\log \big(r\mu_{i}(T_*) \big) \big\} \Big] > - \sum_{i=1}^d \mu_i(0)\log \big\{ -r\mu_{i}(0)\log \big(r\mu_{i}(0)\big) \big\} +\kappa,
	\end{equation*}
	with $\kappa$ in \eqref{Eq: kappa of entropy with LOGLOG}.
	
	\smallskip
	
	We move on to the trading strategy $\psi$ multiplicatively generated by this function \eqref{Def: entropy with LOGLOG}. From \eqref{Eq: psi}, \eqref{Eq: value of MTS} as well as \eqref{Con: partial in entropy with LOGLOG} and \eqref{Eq: gamma of entropy with LOGLOG}, we have
	\begin{equation}	\label{Eq: psi of entropy with LOGLOG}
		\psi_i(t)= -K(t) \Big[ p+\log \big\{ -r\mu_{*i}(t)\log \big(r\mu_{*i}(t)\big) \big\} \Big], \qquad i=1, \cdots, d;
	\end{equation}
	with the value function
	\begin{equation*}
	V^{\psi}(t) =-K(t) \Big[ p + \sum_{i=1}^d \mu_i(t)\log \big\{ -r\mu_{*i}(t)\log \big(r\mu_{*i}(t)\big) \big\} \Big],
	\end{equation*}
	where
	\begin{equation*}
	K(t) :=\exp \bigg(- \sum_{i=1}^{d} \int_0^t \frac{1+\frac{1}{\log \big(r\mu_{*i}(s) \big)} }{p+\sum_{j=1}^d \mu_j(t)\log \big\{ -r\mu_{*j}(t)\log \big(r\mu_{*j}(t)\big) \big\}} d \mu_{*i}(s) \bigg).
	\end{equation*}
	
	\smallskip
	
	As the $\Gamma^G$ is nonincreasing in \eqref{Eq: gamma of entropy with LOGLOG}, we again use Theorem~\ref{thm: MGRA3}. We already have $F\big(\mu(t), \mu_*(t)\big)$ and $\kappa$ defined in \eqref{Eq: F of entropy with LOGLOG} and \eqref{Eq: kappa of entropy with LOGLOG} which satisfy the conditions \textit{(i), (ii)}. Thus, the strong arbitrage condition \eqref{Con: multiplicative arb3} becomes
	\begin{equation*}
	\log (A)
	>
	\frac{1 +\sum_{i=1}^{d} L_r(\mu_{i}(0))}{B},
	\end{equation*}
	where
	\begin{equation*}
		A := -p - \log \Big[ \max_{j=1, \cdots, d} \big\{ -r\mu_{*j}(T_*)\log \big( r\mu_{*j}(T_*) \big) \big\} \Big],
	\end{equation*}
	and
	\begin{equation*}
		B := -p - \log \Big[ \max_{j=1, \cdots, d} \big\{ -r\mu_{j}(0)\log \big( r\mu_{j}(0) \big) \big\} \Big].
	\end{equation*}
	
\end{example}

\bigskip

\bigskip

\bigskip

\section{Empirical results} 
 \label{sec: 7}

We present some empirical results regarding the behavior of additively-generated portfolios in the Section~\ref{sec: 6}, using historical market data. We first analyze the value function $V^{\varphi}(\cdot)$ of these portfolios with respect to the market by decomposing it with generating function $G$ and corresponding Gamma function $\Gamma^G$ in \eqref{Eq: value of ATS}. Especially, we show that all value functions of portfolios in Section~\ref{sec: 6} outperform the market portfolio. Then, we present that the different choice of the parameter $p$, explained in Remark~\ref{Re: Reducing threshold in entropy example}, indeed significantly influences the performance of portfolios.

\bigskip

\subsection{Data description and notation}

\medskip

In order to simulate a perfect `closed market', we construct a ``universe" with $d=1085$ stocks which had been continuously traded during $4528$ consecutive trading days between 2000 January 1st and 2017 December 31st. These $1085$ stocks were chosen from those listed at least once among the constituents of the S\&P 1500 index in this period, and did not undergo mergers, acquisitions, bankruptcies, etc. 

\medskip

\begin{rem}
	\label{re: biased selection of stocks}
	This selection of $1085$ stocks is somewhat biased, in the sense that we are looking ahead into the future at time $t=0$ by blocking out those stocks which will go bankrupt in the future. However, the reason for this biased selection is to keep the number of stocks $d$ constant all the time which is the essential assumption of our `closed' market model. If we compose our portfolio from $d=1500$ stocks included in S\&P 1500 index at the beginning, remove one stock whenever it goes bankrupt, or take in a new stock whenever it is newly added to the index, the number $d$ of stocks in our portfolio fluctuates over time and the generating function $G$ would be discontinuous whenever $d$ changes.
	
	One possible solution to this problem is to consider an `open market'. We first fix the value of $d$, say $d=1500$ at the beginning, keep track of price dynamics of all stocks in the market (which should be composed of more than $d$ stocks, say $D$ stocks with $D>d$), rank them by the order of their market capitalization, and construct our portfolio using the top $d=1500$ stocks among $D$ stocks. In this way we can keep the same number $d$ of companies all the time, but considering ranked market weights always involves a `leakage' issue. As explained in Chapter~4.2, 4.3 of \cite{Fe} and Example~6.2 of \cite{Karatzas:Ruf:2017}, this refers to the loss incurred when we have to sell a stock that has been relegated from top $d$ capitalization index to the lower capitalization index. Even worse, as we want to invest only in the top $d$ companies among $D$ companies in this open market, our trading strategy $\varphi = (\varphi_1, \cdots, \varphi_D)$ should satisfy the equations $\varphi_i(t) = 0$ for $i=1, \cdots, D$ whenever the $i$-th company fails to be included in the top $d$ companies at time $t$. However, we do not know how to construct such trading strategy yet.
	
	Thus, it is not easy to make a perfect empirical model, and we decided to select $d=1085$ stocks in a biased manner which fits better to our theoretic model described in the previous sections.
\end{rem}

\medskip

We obtained daily closing prices and total number of shares of these stocks from the CRSP and Compustat data sets. The data can be found here; \textcolor{blue}{https://wrds-web.wharton.upenn.edu/wrds/}. We used \textsl{R} and \textsl{C++} to program our portfolios.

\smallskip

As we used daily data for $N=4528$ days, we discretized the time horizon as $0=t_0 < t_1 < \cdots, < t_{N-1} = T$. For $\ell \in \{1, 2, \cdots, N \}$, we summarize our notations here;
\begin{enumerate}
	\item $S_i(t_\ell)$ : the capitalization (daily closing price multiplied by total number of shares) of $i$\textsuperscript{th} stock at the end of day $t_\ell$.

	\item $\Sigma(t_\ell) := \sum_{i=1}^d S_i(t_\ell)$ : the total capitalization of $d$ stocks at the end of day $t_\ell$. This quantity also represents dollar value of the market portfolio at the end of day $t_\ell$ with the initial wealth $\Sigma(0)$.

	\item $\mu_i(t_\ell) := \frac{S_i(t_\ell)}{\Sigma(t_\ell)} $ : the $i$\textsuperscript{th} market weight at the end of day $t_\ell$.
	
	\item $\pi_i(t_\ell)$ : the additively generated portfolio weight of the $i$\textsuperscript{th} stock at the end of day $t_\ell$ which can be computed using the equation \eqref{Def: pi of ATS}. Note that $\sum_{i=1}^d \pi_i(t_\ell) = 1$ holds.
	
	\item $W(t_\ell)$ : the total value of the portfolio at the end of day $t_\ell$. Then,  $W(t_\ell)\pi_i(t_\ell)$ represents the amount of money invested by our portfolio in $i$\textsuperscript{th} stock at the end of day $t_\ell$.
\end{enumerate}

As the capitalization of $i$\textsuperscript{th} stock at the beginning of day $t_\ell$ should be equal to $S_i(t_{\ell-1})$, the capitalization of the same stock at the end of the last trading day $t_{\ell-1}$, we also deduce that $\Sigma(t_{\ell-1})$, $\mu_i(t_{\ell-1})$, $\pi_i(t_{\ell-1})$, and $W(t_{\ell-1})$ represent the total capitalization, $i$\textsuperscript{th} market weight, $i$\textsuperscript{th} additively generated portfolio weight, and the money value of portfolio at the beginning of day $t_\ell$, respectively.

\smallskip

The transaction, or rebalancing, of our portfolio on day $t_\ell$, is made at the beginning of the day $t_\ell$, using the market weights $\mu_i(t_{\ell-1})$ at the end of the last trading day. We compute $\pi_i(t_{\ell-1})$ from $\mu_i(t_{\ell-1})$ via \eqref{Def: pi of ATS}, and re-distribute the generated value $W(t_{\ell-1})$ according to the these weights $\pi_i(t_{\ell-1})$. Then, the monetary value of portfolio $W(t_\ell)$ at the end of day $t_\ell$ can be calculated as
\begin{equation*}
	W(t_\ell) = \sum_{i=1}^d W(t_{\ell-1})\pi_i(t_{\ell-1}) \frac{S_i(t_\ell)}{S_i(t_{\ell-1})}.
\end{equation*}
In order to compare the performance of our portfolios with the market portfolio, we set our initial wealth as $W(0) = \Sigma(0)$ and compare the evolutions of $\Sigma(\cdot)$ and $W(\cdot)$. Once the initial amount $W(0)$ invested in our portfolio is determined, the monetary value of the portfolio can be obtained recursively by the above equation. However, $W(\cdot)$ can be defined with the trading strategy $\varphi_i(\cdot)$ in \eqref{Def: varphi} or \eqref{Eq: varphi};
\begin{equation}	\label{Def: money value of portfolio}
	W(\cdot) = \sum_{i=1}^d S_i(\cdot)\varphi_i(\cdot).
\end{equation}
Then, the value $V^{\varphi}(\cdot)$ with respect to the market, defined as in \eqref{Def: value} or represented as in \eqref{Eq: value of ATS}, has another representation as the ratio between the money value of our portfolio and total market capitalization;
\begin{equation*}
	V^{\varphi}(\cdot)
	= \sum_{i=1}^d \varphi_i(\cdot)\mu_i(\cdot) 
	= \sum_{i=1}^d \varphi_i(\cdot)\frac{S_i(\cdot)}{\Sigma(\cdot)}
	= \frac{W(\cdot)}{\Sigma(\cdot)},
\end{equation*}
and the expression `value of trading strategy (or portfolio) with respect to the market' makes sense. Furthermore, the excess return rate $R^{\varphi}(\cdot)$ of the portfolio defined in \eqref{Def: excess return rate} can be represented as
\begin{equation*}
	R^{\varphi}(\cdot)
	= \frac{V^{\varphi}(\cdot) - V^{\varphi}(0)}{V^{\varphi}(0)}
	= \frac{\frac{W(\cdot)}{\Sigma(\cdot)} - 1}{1}
	= \frac{W(\cdot) - \Sigma(\cdot)}{\Sigma(\cdot)}
	~\big(= V^{\varphi}(\cdot) - 1 \big),
\end{equation*} 
and the expression `excess return rate with respect to the market' also makes sense. Here, $V^{\varphi}(0)=1$ because we set $W(0) = \Sigma(0)$. In the last part of following subsection, we show the evolutions of $W(\cdot)$ of several portfolios to compare their performance.

\bigskip

\subsection{Empirical results}

\medskip

We first decompose the value functions $V^{\varphi}(\cdot)$ of trading strategies additively generated from the functions $G$ in entropic examples (Example~\ref{Ex: original entropy}, \ref{Ex: entropy with MAX}, \ref{Ex: entropy with MIN}, and \ref{Ex: entropy with LOGLOG}) into the generating function $G\big(\mu(\cdot), A(\cdot)\big)$ and the corresponding Gamma function $\Gamma^G(\cdot)$. For easy comparison, we normalized all generating functions so that $G\big(\mu(0), A(0)\big)=1$ holds, and shifted up the Gamma functions by $1$ in Figure~\ref{fig: decomposition}.

\begin{figure}[h]
	\centering
	\caption{Decomposition of value function of additively generated trading strategies}
	\label{fig: decomposition}
	\begin{subfigure}{.5\textwidth}
		\centering
		\includegraphics[width=1\linewidth]{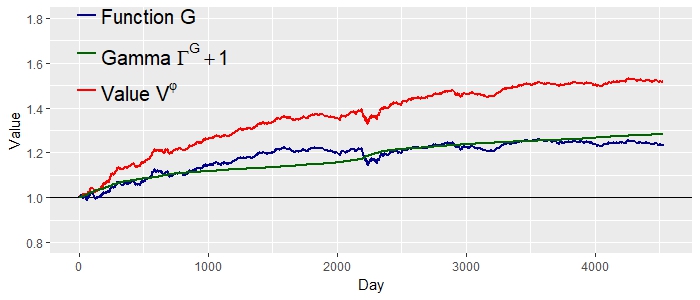}
		\captionsetup{justification=centering}
		\caption{Example~\ref{Ex: original entropy} \\ Original entropy, $p=9$}
		\label{fig: sub1}
	\end{subfigure}%
	\begin{subfigure}{.5\textwidth}
		\centering
		\includegraphics[width=1\linewidth]{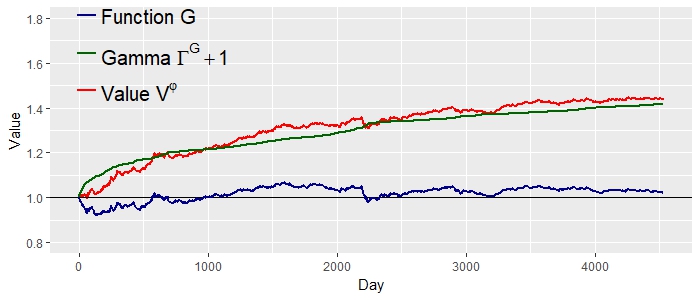}
		\captionsetup{justification=centering}
		\caption{Example~\ref{Ex: entropy with MAX} \\ Entropy with running maximum, $p=9$}
		\label{fig: sub2}
	\end{subfigure}
	\begin{subfigure}{.5\textwidth}
		\centering
		\includegraphics[width=1\linewidth]{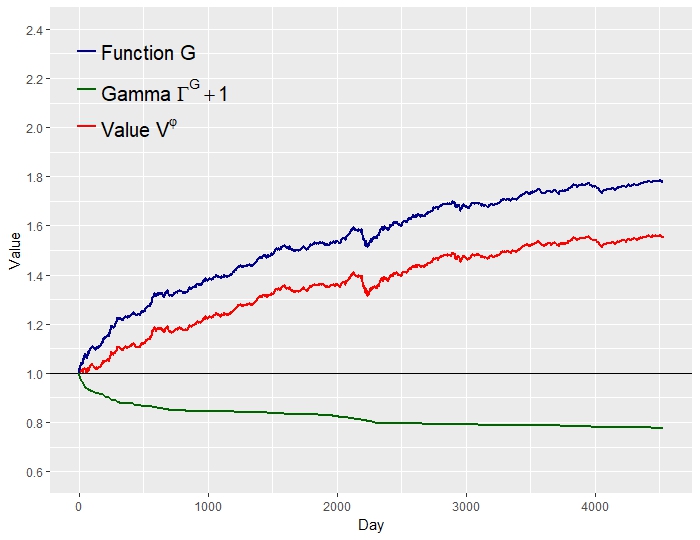}
		\captionsetup{justification=centering}
		\caption{Example~\ref{Ex: entropy with MIN} \\ Entropy with running minimum, $p=9$}
		\label{fig: sub3}
	\end{subfigure}%
	\begin{subfigure}{.5\textwidth}
		\centering
		\includegraphics[width=1\linewidth]{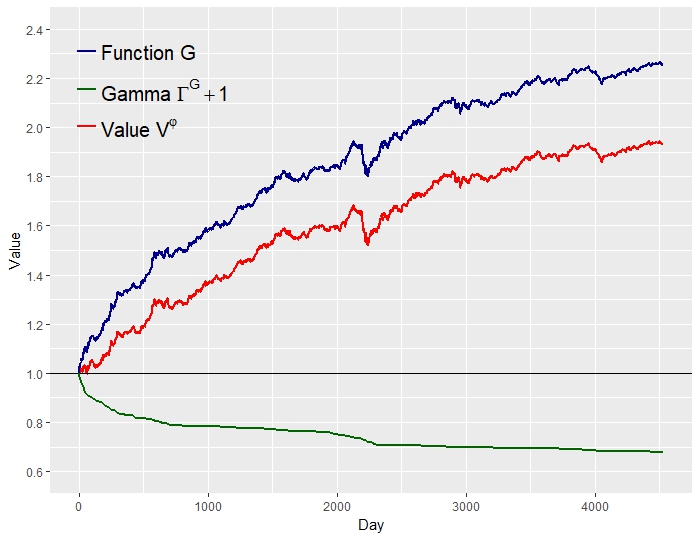}
		\captionsetup{justification=centering}
		\caption{Example~\ref{Ex: entropy with LOGLOG} \\ Iterated entropy with running minimum, $p=9$, $r=5$}
		\label{fig: sub4}
	\end{subfigure}
\end{figure}

\smallskip

Figure~\ref{fig: decomposition} confirms that all trading strategies additively generated in Section~\ref{sec: 6} outperform the market as the values $V^{\varphi}$ (red lines in the figure) gradually increase. In sub-figures (a) and (b), the growth of the value $V^{\varphi}$ comes from the growth of the Gamma function. In contrast, even though the Gamma function decreases, the value of trading strategy grows as the function $G$ increases substantially in sub-figures (c) and (d). In the sub-figure (d), we set the parameter $r=5$ as it is the largest integer satisfying the equation \eqref{Con: initial condition in entropy with LOGLOG}; initial market weights data give us $\max_i \mu_i(0) = 0.065$ and $0.065 < 1/(5e)$ holds. We chose the same parameter $p=9$ (See Remark~\ref{Re: Reducing threshold in entropy example}) in all sub-figures for fair comparison, but this is a very sloppy choice of the parameter $p$ for (a), (b), and (c). If we chose the value of $p$ using statistical estimation elaborately in each examples, the performance of portfolio would be improved, as Figure~\ref{fig: different p} represents in the case of Example~\ref{Ex: original entropy}.

\begin{figure}[h!]
	\centering
	\captionsetup{justification=centering}
	\caption{Value of additively generated trading strategies from Example~\ref{Ex: original entropy} with different $p$ values}
	\label{fig: different p}
	\includegraphics[width=0.7\linewidth]{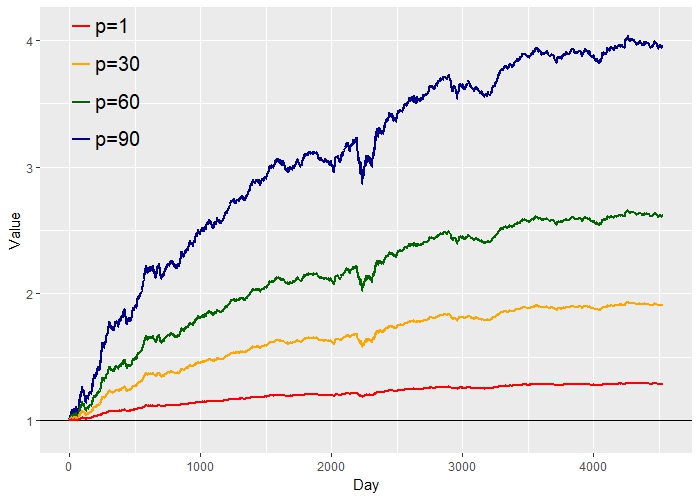}
\end{figure}

\smallskip

Figure~\ref{fig: different p} shows the values of additively generated portfolios in Example~\ref{Ex: original entropy} with different choices of the parameter $p$. We can verify that trading strategy with bigger value of $p$ performs better as described in Remark~\ref{Re: Reducing threshold in entropy example}. From the data, the Gibbs entropy $-\sum_{i=1}^{1085} \mu_i(t) \log \mu_i(t)$ of market weights ranged from $4.954$ to $5.726$ during 4528 days. Thus, $p=90$ is a safe estimation of the parameter $p$ which guarantees the non-negativity of the function $G$ in \eqref{Eq: original entropy} as $\log 90 < 4.5 < 4.954$ holds.

\smallskip

Finally, `dollar values' $W(\cdot)$ of portfolios from four examples of Section~\ref{sec: 6}, which are defined as in \eqref{Def: money value of portfolio}, along with the total market value $\Sigma(\cdot)$ of $d=1085$ stocks from the start of 2000 to the end of 2017, are illustrated in Figure~\ref{fig: OVERALL}. Dollar values are normalized by replacing $W(\cdot)$ by $W(\cdot)/W(0)$. In Figure~\ref{fig: OVERALL}, while the market capitalization had been approximately doubled during $18$ years, the dollar values of all other portfolios had been grown more than $4.5$ times. Parameters are appropriately chosen using statistical estimation in each portfolio.

\begin{figure}[h!]
	\centering
	\captionsetup{justification=centering}
	\caption{(Normalized) Dollar values of portfolios over 18 years}
	\label{fig: OVERALL}
	\includegraphics[width=.7\linewidth]{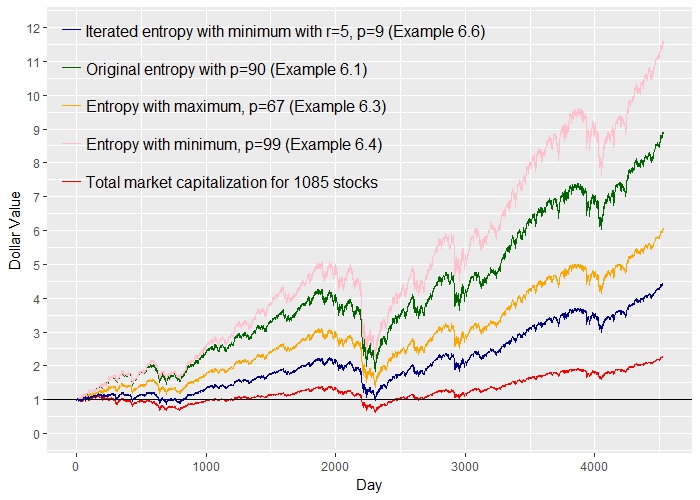}
\end{figure}

\bigskip

\bigskip

\bigskip

\section{Conclusion}
 \label{sec: 8}

\cite{Karatzas:Ruf:2017} introduced an alternative ``additive" way of functional generation of trading strategies and compared it to the original ``multiplicative'' way of E.R. Fernholz. This new approach weakens the assumption on the asset prices from It\^o processes to continuous semimartingales, characterizes the class of functions called ``Lyapunov functions" which generate trading strategies leading to strong arbitrage with respect to the market, and gives a very simple sufficient condition for strong arbitrage. The present paper takes more generalized approaches to these two ways of functional generation. The results of this paper can be summarized as follows:

\begin{enumerate}
	\item We show how to generate both additively and multiplicatively, trading strategies without any probabilistic assumptions on the market model. This is done by using the celebrated pathwise It\^o calculus, and the only analytic assumption we impose is that the market weights admit continuous covariations in a pathwise sense. In the practical sense, we do not have to care about this analytic assumption because market weights data are given as the form of discrete time-series and such data always admit pathwise covariations. 
	
	\item We extend the class of functions which generate trading strategies by introducing an additional argument of finite variation other than market weights as the input. Inserting this argument in the generating function gives extra flexibility in portfolio construction and this has been dealt with in other literatures. However, we present some new examples of such extra argument which gives us simple sufficient condition leading to strong arbitrage relative to the market.
	
	\item We also extend the class of functions which generate additive and multiplicative strong relative arbitrage by giving new sufficient conditions. The new conditions allow the function not be ``Lyapunov'', or concave with respect to the market weights, in order to generate strong relative arbitrage or to outperform the market portfolio in the long run. We also present empirical results of portfolios which indeed outperform the market.
	
	\item We further extend the class of portfolio-generating-functions from twice-differentiable to less smoother, namely absolutely continuous, functions with help of the pathwise Tanaka formula. Using Tanaka formula involves the concept of local times and this yields new interesting types of portfolios and corresponding strong arbitrage conditions.
\end{enumerate}

\medskip

While this paper generalizes the functional generation of portfolios in several respects, we suggest some new questions. First, this paper assumes a `closed market', in other words, the number of stocks $d$ is fixed. In this respect, it fails to represent or resemble the real market. As explained in Remark~\ref{re: biased selection of stocks}, an `open market' models the real world better, but nothing seems to be known on how to construct trading strategies in this open market. Secondly, the market weights in this paper should have finite second variation along a sequence of time partitions; can something be said, along the lives of \cite{Cont_Perkowski}, regarding price dynamics, or market weights, with finite $p$-th variation for $p>2$?

\bigskip

\bigskip

\bigskip

\begin{appendices}
\section{Proofs}
\label{sec: appendix}
\begin{proof} [Proof of Theorem~\ref{Thm : Ito formula}]
	Using the telescoping sum representation, we obtain
	\begin{align}
	f\big(X(t), A(t)\big) - f\big(X(0)), A(0)\big)
	&= \sum_{\substack{[t_j, t_{j+1}] \in \mathbb{T}_n \\ t_j \leq t}} \Big\{ f\big(X(t_{j+1}), A(t_{j+1})\big) - f\big(X(t_{j}), A(t_{j})\big) \Big\}			\nonumber
	\\
	&= \sum_{\substack{[t_j, t_{j+1}] \in \mathbb{T}_n \\ t_j \leq t}} \Big\{ f\big(X(t_{j+1}), A(t_{j+1})\big) - f\big(X(t_{j+1}), A(t_{j})\big) \Big\}			\label{Eq: xaa}
	\\
	&+ \sum_{\substack{[t_j, t_{j+1}] \in \mathbb{T}_n \\ t_j \leq t}} \Big\{ f\big(X(t_{j+1}), A(t_{j})\big) - f\big(X(t_{j}), A(t_{j})\big) \Big\}.				\label{Eq: xxa}
	\end{align}
	The Taylor expansion, applied to the components of the function $A$ in the sum \eqref{Eq: xaa}, gives
	\begin{align}
	&~~~~~\sum_{\substack{[t_j, t_{j+1}] \in \mathbb{T}_n \\ t_j \leq t}} \Big\{ f\big(X(t_{j+1}), A(t_{j+1})\big) - f\big(X(t_{j+1}), A(t_{j})\big) \Big\}		\label{Eq: xaa2}
	\\
	&=\sum_{\substack{[t_j, t_{j+1}] \in \mathbb{T}_n \\ t_j \leq t}} \sum_{\ell=1}^m D_{\ell}f\big(X(t_{j+1}), A(t_{j})\big)\big(A_{\ell}(t_{j+1}) - A_{\ell}(t_{j})\big)
	+\sum_{\substack{[t_j, t_{j+1}] \in \mathbb{T}_n \\ t_j \leq t}} \sum_{\ell=1}^m r\big(A_{\ell}(t_{j+1}) - A_{\ell}(t_{j})\big),				\nonumber
	\end{align}
	where the last remainder term is bounded by
	\begin{equation*}
	r\big(A_{\ell}(t_{j+1}) - A_{\ell}(t_{j})\big) \leq \phi\Big(\max_{t_j}\big|A_{\ell}(t_{j+1})-A_{\ell}(t_{j})\big|\Big) \Big|A_{\ell}(t_{j+1})-A_{\ell}(t_{j})\Big|
	\end{equation*}
	for some function $\phi$ with the property $\lim_{x \rightarrow 0} \phi(x) = 0$. Since $A$ is continuous and of bounded variation, the last double sum of the right-hand side of \eqref{Eq: xaa2} goes to zero as $n \rightarrow \infty$ and the sum \eqref{Eq: xaa} converges to the Lebesgue-Stieltjes integral
	\begin{equation*}
	\sum_{\ell=1}^m \int_0^t D_{\ell}f\big(X(s), A(s)\big) dA_{\ell}(s),
	\end{equation*}
	as $n \rightarrow \infty$. On the other hand, again by the Taylor expansion applied to the components of the function $X$ in the sum \eqref{Eq: xxa}, we obtain
	\begin{align}
	&~~~~~\sum_{\substack{[t_j, t_{j+1}] \in \mathbb{T}_n \\ t_j \leq t}} \Big\{ f\big(X(t_{j+1}), A(t_{j})\big) - f\big(X(t_{j}), A(t_{j})\big) \Big\}		\nonumber
	\\
	&=\sum_{\substack{[t_j, t_{j+1}] \in \mathbb{T}_n \\ t_j \leq t}} \sum_{i=1}^d \partial_i f\big(X(t_{j}), A(t_{j})\big)\big(X_{i}(t_{j+1}) - X_{i}(t_{j})\big)		\label{Eq: fxx}
	\\
	&~~~+\frac{1}{2} \sum_{\substack{[t_j, t_{j+1}] \in \mathbb{T}_n \\ t_j \leq t}} \sum_{i, k=1}^d \partial^2_{i, k} f \big( X(t_j), A(t_j) \big) \big(X_i(t_{j+1})-X_i(t_j)\big) \big(X_k(t_{j+1})-X_k(t_j)\big)												\label{Eq: fxxxx}
	\\
	&~~~+\sum_{\substack{[t_j, t_{j+1}] \in \mathbb{T}_n \\ t_j \leq t}} \sum_{i, k=1}^d
	r\big(X_{i, k}(t_{j+1}) - X_{i, k}(t_{j})\big),								\label{Eq: rx}
	\end{align}
	where the last remainder term \eqref{Eq: rx} is bounded by
	\begin{equation*}
	r\big(X_{i, k}(t_{j+1}) - X_{i, k}(t_{j})\big) \leq \psi\Big(\max_{t_j, i, k}\big|X_{i, k}(t_{j+1})-X_{i, k}(t_{j})\big|\Big) \big(X_{i}(t_{j+1})-X_{i}(t_{j})\big) \big(X_{k}(t_{j+1})-X_{k}(t_{j})\big),
	\end{equation*}
	for some function $\psi$ with the property $\lim_{x \rightarrow 0} \psi(x) = 0$. Again, by the continuity of $X$ and by the fact that $X$ admits the pathwise quadratic covariation in the sense of \eqref{Def: quadratic}, the double sum \eqref{Eq: rx} approaches zero as $n \rightarrow \infty$. The sum \eqref{Eq: fxxxx} converges to the Lebesgue-Stieltjes integral
	\begin{equation*}
	\frac{1}{2} \sum_{i,k=1}^d \int_0^t \partial^2_{i, k}  f\big(X(s), A(s)\big)d\langle X_i, X_k \rangle(s),
	\end{equation*}
	again by the existence of the pathwise quadratic covariation of $X$. As all the other terms converge, the remaining sum \eqref{Eq: fxx} should converge to some limit, which we call `F{\"o}llmer-It\^o integral' as in \eqref{Eq: F-I integral}.
\end{proof}

\bigskip

\begin{proof}	[Proof of Theorem~\ref{thm: tanaka formula}]
	For any two real numbers $a$ and $b$, by applying the integration by parts formula with the notation \eqref{Def: parenthesis}, we have the equation
	\begin{align}
	f(b) - f(a) &= \int_a^b f'(x) dx 			\nonumber
	\\
	&=\begin{cases} 			\nonumber
	~~~\int_a^b f'(x) (b-x)^0 dx
	= -f'(x) (b-x)\big|^b_{x=a} + \int_{(a, b]} (b-x)df'(x), ~\quad \text{if}~~~ a \leq b
	\\ 			\nonumber
	\\
	-\int_b^a f'(x) (b-x)^0 dx
	= f'(x) (b-x)\big|^a_{x=b} - \int_{(b, a]} (b-x)df'(x), \qquad \text{if}~~~ b < a
	\end{cases}
	\\ 			\nonumber
	\\ 			\nonumber
	&=\begin{cases} 		 			\nonumber
	f'(x) (b-a) + \int_{(a, b]} (b-x)df'(x), \qquad \text{if}~~~a \leq b
	\\
	\\
	f'(x) (b-a) - \int_{(b, a]} (b-x)df'(x), \qquad \text{if}~~~b < a
	\end{cases}
	\\
	&= f'(x) (b-a) + \int_{\mathbb{R}}\mathbbm{1}_{\llparenthesis a, b \rrbracket}(x) |b-x|df'(x).			\label{Eq : integration by parts}
	\end{align}
	Thus, using the telescoping sum
	\begin{equation*}
	f(X_t)-f(X_0) = \sum_{\substack{[t_j, t_{j+1}] \in \mathbb{T}_n \\ t_j \leq t}} \big( f(X_{t_{j+1}}) - f(X_{t_{j}}) \big)
	\end{equation*}
	for the sequence of partitions $\mathbb{T} = (\mathbb{T}_n)_{n \in \mathbb{N}}$ of $[0, T]$, the above equality becomes
	\begin{align}
	f(X_t) - f(X_0) &- \sum_{\substack{[t_j, t_{j+1}] \in \mathbb{T}_n \\ t_j \leq t}} f'(X_{t_j}) (X_{t_{j+1}} - X_{t_j }) 		\label{Eq : change of variable}
	\\
	&= \sum_{\substack{[t_j, t_{j+1}] \in \mathbb{T}_n \\ t_j \leq t}} \int_{\mathbb{R}}\mathbbm{1}_{\llparenthesis X_{t_{j}}, X_{t_{j+1}} \rrbracket}(x) |X_{t_{j+1}}-x|df'(x) 		\nonumber
	= \int_\mathbb{R}L_t^{X, \mathbb{T}_n}(x)df'(x),	\nonumber		
	\end{align}
	thanks to the definition \eqref{Def: discrete local time}. The last integral on the right-hand side of \eqref{Eq : change of variable} converges to the last integral of \eqref{Eq : tanaka formula}, since $L_t^{X, \mathbb{T}_n}$ converges uniformly to $L_t$, and the result follows.
\end{proof}

\end{appendices}

\newpage

\bibliography{aa_bib}
\bibliographystyle{apalike}

\end{document}